\documentclass[aps,pra,reprint,superscriptaddress]{revtex4-2}

\usepackage{amsmath,amssymb,amsthm,hyperref}
\usepackage{graphicx}
\usepackage{braket}
\usepackage{tikz}
\usetikzlibrary{calc,fit,positioning,arrows.meta}

\newtheorem{theorem}{Theorem}
\newtheorem{lemma}{Lemma}
\newtheorem{corollary}{Corollary}
\newtheorem{remark}{Remark}
\newtheorem{proposition}{Proposition}

\newcommand{\Tr}{\mathrm{Tr}}
\newcommand{\bbE}{\mathbb{E}}
\newcommand{\bbP}{\mathbb{P}}
\newcommand{\kB}{k_{\mathrm{B}}}
\newcommand{\Id}{\mathbb{I}}

\begin{document}
\title{Battery-Explicit Thermodynamic Witnesses of Bell Post-Quantumness}
\author{Piotr {\'C}wikli{\'n}ski}
\email{piotr.cwiklinski@ug.edu.pl}
\affiliation{International Centre for Theory of Quantum Technologies, University of Gda{\'n}sk, ul. prof. Marii Janion 7, 80-309 Gda\'nsk, Poland}
\date{\today}

\begin{abstract}
We introduce a battery-explicit thermodynamic witness of post-quantum Bell correlations. In each round, a single supplied excitation is routed into an explicit two-level battery if and only if a Bell-game condition is satisfied. The routing operation is implemented by an energy-preserving controlled SWAP, with all logical control registers taken to be degenerate. Thus the correlation resource does not create energy; it only determines the probability that the supplied excitation reaches the battery.

The construction is first formulated for finite two-player XOR games. For any such game, the mean battery charge is exactly the game success probability multiplied by the battery gap. Optimizing over local, quantum, or nonsignalling behaviours therefore turns the corresponding game values into local, quantum, or nonsignalling battery-charge ceilings. For the CHSH game, Tsirelson's bound becomes a strict quantum ceiling on the mean battery charge, while a PR-box behaviour reaches the single-excitation cap.

The witness is trusted-module rather than device-independent: it assumes calibrated Hamiltonians, correct classical wiring, and a trusted energy-preserving battery module. We also discuss a reversible-controller implementation, finite-statistics certification from work data, robustness to imperfect battery readout, and cyclic bookkeeping showing that no positive net work is obtained once fuel restoration and memory erasure are included.
\end{abstract}

\maketitle

\section{Introduction}
\label{sec:introduction}

Bell inequalities show that different physical theories allow different
performances in distributed information-processing tasks.  The CHSH inequality
is the standard example: local behaviours obey the classical bound, quantum
behaviours are limited by Tsirelson's bound, and nonsignalling post-quantum
behaviours such as PR boxes can reach the algebraic value without allowing
communication \cite{CHSH1969,Tsirelson1980,PopescuRohrlich1994,BarrettPRA2005}.
More generally, two-player XOR games provide a compact language for Bell-type
tasks with binary outputs \cite{CleveHoyerTonerWatrous2004}.  They include CHSH
and the chained Bell games as important examples
\cite{BraunsteinCaves1990,Wehner2006}.

A separate line of work in quantum thermodynamics emphasizes that work,
information, and control must be accounted for explicitly.  In resource-theoretic
formulations, thermodynamic transformations are described by specifying allowed
operations, nonequilibrium resources,
\cite{HorodeckiNatComm2013,BrandaoPNAS2015,LostaglioNatComm2015,PerryPRX2018, Renes2014, Lostaglio2019RoPP, Faist2015, Mazurek_2018, LostaglioArxiv1410, ETO2018, CaravelliQuantum2021, Su2025_thermoresource, HackMendl2025, KorzekwaWorkFromCoh2016, HuDing2019, SonNg2023, Shiraishi2025, Cwiklinski2015,Gour2015, LostaglioReview2019, ChitambarGour2019}, batteries, and  memories \cite{Rodriguez_2024_Optimal_Control, Malavazi2025ChargepreservingOI, Borhan2025_batteries, LipkaBartosik2021secondlawof, CaravelliQuantum2021, Zhi2-25_quantumbatteries, chaki2025positivenonpositivemeasurementsenergy, Chaki2026_assistantsquantumbatteries, Chaki2025_distillationfromquantumbatteries}.
Likewise, Maxwell-demon and feedback-control settings show that information can
modify operational work balances, provided that memories and erasure costs are
included \cite{Landauer1961,Bennett1982,SagawaUeda2008,SagawaUeda2010,Parrondo2015,ReebWolf2014,Goold2016}.
These considerations are especially important whenever a controller or a
battery is used: if it is not modelled explicitly, an apparent thermodynamic
advantage may come from an unaccounted work source.

Thermodynamic observables have also been proposed as witnesses of quantum
properties. In particular, heat exchange with a thermal environment, assisted
by a quantum memory, can witness resources such as entanglement and coherence
\cite{deOliveiraJuniorBraskLipkaBartosik2025}.  The present work follows a
different route.  We do not use heat generation or thermalization as the
witnessing signal.  Instead, a pre-existing excitation is routed into an
explicit battery by an energy-preserving operation, and the mean battery charge
is compared with Bell-game resource-class bounds.

In this paper we connect these two viewpoints in a deliberately conservative
way.  We construct a trusted battery witness of post-quantum Bell correlations.
The witness does not extract work from Bell correlations and does not derive
Tsirelson's bound from thermodynamics.  Instead, it implements a calibrated
energy-preserving transducer: a single supplied excitation is routed into an
explicit two-level battery if and only if a Bell-game winning condition is
satisfied.

The construction is formulated first for finite two-player XOR games.  Let
\[
\mathcal{G}=(\mathcal{U},\mathcal{V},\pi,f)
\]
be such a game, and let \(P(a,b|u,v)\) be the behaviour generated by the
correlation resource.  The referee samples \((U,V)\sim\pi\) and an independent
uniform bit \(R\), and defines
\[
X=f(U,V)\oplus R.
\]
The players' answers define
\[
G=A\oplus B\oplus R.
\]
Then
\[
G=X
\quad\Longleftrightarrow\quad
A\oplus B=f(U,V).
\]
Thus equality of the two logical bits is exactly the XOR-game winning event.
The random pad \(R\) makes the individual bits unbiased; the game performance is
stored in their correlation.

The thermodynamic module consists of two degenerate logical registers storing
\(X\) and \(G\), a fuel qubit \(F\), and a battery qubit \(W\).  The fuel and
battery have the same energy gap \(\Delta\).  At the beginning of a round the
fuel is excited and the battery is empty.  If the equality condition \(X=G\) is
satisfied, an equality-controlled SWAP transfers the excitation from \(F\) to
\(W\); otherwise the fuel-battery state is left unchanged.  Since the SWAP only
exchanges equal-energy states, and the logical registers are degenerate, the
operation is energy preserving.

We denote by \(W_{\rm bat}\) the increase of the battery energy in one round.
In the ideal protocol this random variable takes only the two values
\[
W_{\rm bat}=0
\quad\text{or}\quad
W_{\rm bat}=\Delta .
\]
It is equal to \(\Delta\) precisely on those rounds in which the game is won.

Taking the average over many rounds gives the basic work--game identity
\begin{equation}
\bbE[W_{\rm bat}]
=
\Delta\,p_{\rm succ}^{\mathcal{G}}(P),
\label{eq:intro_work_success}
\end{equation}
where \(p_{\rm succ}^{\mathcal{G}}(P)\) is the success probability of the
behaviour \(P\) in the game \(\mathcal G\).  Thus the mean battery charge is not
an additional thermodynamic quantity independent of the Bell experiment.  It is
the game success probability written in energy units.

Now fix a class \(\mathsf C\) of admissible behaviours, for instance the local,
quantum, or nonsignalling class.  Let
\[
\omega_{\mathsf C}(\mathcal G)
\]
denote the optimal success probability of the game within that class.  Equation
\eqref{eq:intro_work_success} immediately turns this game value into a battery
ceiling,
\begin{equation}
W_{\mathsf C}^{\max}(\mathcal G)
=
\Delta\,\omega_{\mathsf C}(\mathcal G).
\label{eq:intro_class_ceiling}
\end{equation}
An observed mean battery charge above the quantum ceiling of the game is then a
trusted-module witness that the effective behaviour is not quantum-realizable.

For the CHSH game, let \(S(P)\) be the standard CHSH expression of the
behaviour \(P\).  Since the CHSH winning probability is
\[
p_{\rm win}(P)=\frac12+\frac{S(P)}8,
\]
Eq.~\eqref{eq:intro_work_success} gives
\begin{equation}
\frac{\bbE[W_{\rm bat}]}{\Delta}
=
\frac12+\frac{S(P)}8.
\label{eq:intro_chsh_work}
\end{equation}

The algebraic part of the construction is simple by design: a Boolean winning
predicate is converted into a binary control signal.  The point of the present
work is not that this Boolean computation is difficult, but that it can be
realized as an explicit energy-preserving battery operation with no hidden work
source in the controller.  The contribution is the calibrated thermodynamic
transduction of Bell-game performance into a battery degree of freedom, together
with the associated resource-class thresholds, finite-statistics witness, and
cyclic bookkeeping showing that no positive net work is produced.

The result should be read with its assumptions visible.  The witness is not a
loophole-free Bell test and is not device independent.  It assumes calibrated
Hamiltonians, correct classical wiring, a trusted equality-controlled SWAP, and
a calibrated battery readout.  It also does not produce positive cyclic work.
The supplied excitation is the energetic resource.  The correlation resource
only determines whether that excitation reaches the battery.  If the device is
used cyclically, consumed fuel excitations must be restored, and any persistent
memory record must be reset.

Recent work has also studied work-extraction tasks in which quantum structure
itself gives an advantage over classical commuting implementations.  For
example, Ref.~\cite{RoutRavichandranHorodeckiChaturvedi2026} shows that
incompatible Hamiltonian settings can exceed classical commuting limits in an
average-work task.  Our setting is different: the Hamiltonians of the fuel and
battery are fixed and commuting, and the nonclassicality appears through the
correlation resource controlling an energy-preserving routing operation.

The paper is organized as follows.  Section~\ref{sec:physical_setting} states
the physical setting and trusted-module assumptions.
Section~\ref{sec:battery_passivity_charging} discusses stored battery energy,
passivity, ergotropy, coherent excitation transfer, and dissipative corrections.
Section~\ref{sec:xor_transducer} introduces the XOR-game battery transducer and
proves the main identity.  Section~\ref{sec:examples} applies the construction
to CHSH and chained Bell games.  Section~\ref{sec:controller_bookkeeping}
discusses reversible controller implementations and cyclic thermodynamic
bookkeeping.  Section~\ref{sec:statistics_robustness} treats finite-statistics
certification and imperfect battery readout.  Section~\ref{sec:discussion}
summarizes the scope and limitations.

\section{Physical setting and trusted-module assumptions}
\label{sec:physical_setting}

We now specify the physical model used throughout the paper.  The aim is to
isolate the thermodynamic part of the construction from the source of
correlations.  The correlation resource may be local, quantum, nonsignalling, or
post-quantum.  The energy module is trusted and has explicitly specified
Hamiltonians. The overall architecture of the trusted battery-witness protocol is summarized in Fig.~\ref{fig:protocol_scheme}.

\subsection{Correlation resource and game data}

The correlation resource is used only during the distributed game stage.  The
referee samples questions \((u,v)\), sends \(u\) to Alice and \(v\) to Bob, and
receives output bits \(a,b\).  This defines a behaviour
\[
P(a,b|u,v).
\]
No thermodynamic assumption is made about the physical device that produces this
behaviour.  In particular, the battery module does not require a quantum model
of the correlation resource.

After the game outputs are produced, the relevant classical data are available
to the trusted local controller that operates the battery module.  This is a
trusted-module setting: the thermodynamic stage is not required to be spacelike
separated and is not itself a Bell test.

\subsection{Logical registers}

The logical registers used to store the target bit, the guess bit, and any
controller memory are taken to be degenerate.  Thus their Hamiltonian is
idealized as
\[
H_{\rm logical}=0.
\]
This assumption is used only for bookkeeping: it ensures that logical operations
on these registers do not themselves raise or lower the explicit energy of the
fuel-battery system.

This is an idealization.  Physically, it means that the logical energy splittings
are negligible compared with the battery gap \(\Delta\), or that they are
compensated by a calibrated control system that is not being counted as a hidden
work source.

\subsection{Fuel and battery}

The energetic resource in each round is one excitation in a fuel qubit \(F\).
The fuel Hamiltonian is
\[
H_F=\Delta\ket{1}\!\bra{1}_F.
\]
The battery qubit \(W\) has the same gap,
\[
H_W=\Delta\ket{1}\!\bra{1}_W.
\]
At the beginning of a round,
\[
F=\ket{1}_F,
\qquad
W=\ket{0}_W.
\]
Thus the initial fuel-battery state contains one transferable excitation.

The one-round battery-energy increase introduced in the Introduction is defined
more explicitly as
\[
W_{\rm bat}
:=
\Tr(H_W\rho'_W)-\Tr(H_W\rho_W),
\]
where \(\rho_W\) and \(\rho'_W\) are the battery states before and after the
routing operation.  In the ideal protocol,
\[
W_{\rm bat}\in\{0,\Delta\}.
\]

\subsection{Energy-preserving routing}

The only nontrivial energetic operation is a controlled SWAP between \(F\) and
\(W\).  When the control condition is satisfied, the operation maps
\[
\ket{1}_F\ket{0}_W
\longmapsto
\ket{0}_F\ket{1}_W.
\]
When the condition is not satisfied, it acts as the identity.  Since
\[
\ket{1}_F\ket{0}_W
\quad\text{and}\quad
\ket{0}_F\ket{1}_W
\]
have the same total energy \(\Delta\), the SWAP is energy preserving.  The
transducer therefore routes an already supplied excitation; it does not create
one.

The routing operation is compatible with the Thermal Operations paradigm in the
trivial-bath sense: the core step is a global energy-preserving unitary.  A
thermal bath is needed only if one includes irreversible memory reset in a
cyclic implementation.

\subsection{Trusted-module interpretation}

The witness assumes that the following ingredients are calibrated and trusted:
the Hamiltonians \(H_F\) and \(H_W\), the battery gap \(\Delta\), the logical
wiring that computes the control condition, the controlled SWAP, and the
battery readout.  Under these assumptions, a battery charge above the quantum
ceiling of a game implies that the effective behaviour is not
quantum-realizable.

Equivalently, in an actual experiment, observing such a violation means that at
least one of two things is true: either the effective behaviour is outside the
quantum set, or one of the trusted-module assumptions has failed.  The witness
therefore certifies post-quantumness only relative to independent validation of
the energy module and readout.

\subsection{Thermodynamic boundary of the cycle}

The single-round transducer is not a heat engine.  It does not convert heat into
work, and it does not convert Bell correlations into net work.  It only transfers
a supplied excitation to the battery on winning rounds.

If the device is used cyclically, then successful rounds leave the fuel in
\(\ket{0}_F\), and the fuel excitation must be restored before reuse.  This
restoration costs at least \(\Delta\) on those rounds.  Failed rounds leave the
fuel excitation unused and do not require re-excitation of the same fuel qubit.

If the controller stores a persistent success/failure record, that memory must
also be reset before the next cycle.  For blind erasure, Landauer's principle
gives the corresponding erasure cost.  If the full game transcript is stored
inside the thermodynamic cycle, then its reset cost must also be included.  In
the main cyclic bookkeeping below, the transcript registers are treated as
external game data, while the local fuel, battery, and possible success/failure
memory are included explicitly.

\begin{figure*}[t]
\centering
\begin{tikzpicture}[
    box/.style={
        draw,
        rounded corners,
        align=center,
        text width=3.45cm,
        minimum height=1.25cm
    },
    arr/.style={-{Latex[length=2mm]}, thick}
]

\node[box] (ref) at (0,0) {Trusted referee/controller\\
samples \(U,V,R\)\\
\(X=f(U,V)\oplus R\)};

\node[box] (game) at (5.9,0) {XOR-game resource\\
inputs \(u,v\)\\
outputs \(a,b\)};

\node[box] (logic) at (5.9,-2.8) {Classical wiring\\
\(G=A\oplus B\oplus R\)\\
test \(G=X\)};

\node[box] (energy) at (0,-2.8) {Energy module\\
fuel \(F=\ket1\)\\
battery \(W=\ket0\)};

\node[box] (out) at (5.9,-5.6) {Battery output\\
charged iff \(G=X\)\\
equivalently iff game is won};

\coordinate (energyctrl) at ([yshift=4pt]energy.east);
\coordinate (energyswap) at ([xshift=-10pt]energy.south);

\draw[arr] (ref.east) -- node[above] {\(u,v\)} (game.west);
\draw[arr] (game.south) -- node[right] {\(a,b\)} (logic.north);

\coordinate (XRmid) at ($(ref.south)+(0,-0.75)$); \draw[arr] (ref.south) -- node[left,pos=0.45] {\(X,R\)} (XRmid) -| (logic.north);

\draw[arr] (logic.west) -- node[above] {control} (energyctrl);

\draw[arr] (energyswap) |- node[pos=0.22,right,align=center] {controlled\\SWAP} (out.west);

\node[
    draw,
    dashed,
    rounded corners,
    fit=(ref)(logic)(energy)(out)(game),
    inner sep=0.4cm,
    label={[align=center]above:battery-witness protocol}
] {};

\end{tikzpicture}
\caption{
Schematic of the trusted battery witness.  The correlation resource supplies
only the classical outputs \(a,b\) for the XOR game.  The trusted
referee/controller forms the target bit \(X=f(U,V)\oplus R\) and the guess bit
\(G=A\oplus B\oplus R\).  The fuel excitation is routed into the battery by an
energy-preserving controlled SWAP if and only if \(G=X\), which is equivalent to
winning the game.  The energetic resource is the initially supplied fuel
excitation, not the correlation resource.
}
\label{fig:protocol_scheme}
\end{figure*}

\section{Battery energy, passivity, and coherent charging}
\label{sec:battery_passivity_charging}
\subsection{Stored energy, passivity, and ergotropy}
\label{subsec:stored_energy_ergotropy}

The battery variable used in the witness is the increase of the energy of an
explicit two-level work-storage system.  Since in the quantum-battery literature
one often distinguishes stored energy from ergotropy, we now spell out this
distinction for the present model.

The battery Hamiltonian is
\begin{equation}
H_W=\Delta\ket{1}\!\bra{1}_W.
\label{eq:battery_hamiltonian_ergotropy_section}
\end{equation}
In one round of the ideal protocol, the battery is either left in \(\ket0_W\) or
charged to \(\ket1_W\).  Thus the run-resolved battery-energy increase is
\begin{equation}
W_{\rm bat}\in\{0,\Delta\}.
\label{eq:run_resolved_battery_energy}
\end{equation}
This is the quantity entering the witness.

If the charge record is kept, then the interpretation is simple.  Conditional on
a charged round, the battery state is \(\ket1_W\).  This state has ergotropy
\(\Delta\), since the unitary flip \(\ket1_W\mapsto\ket0_W\) lowers the battery
energy by \(\Delta\).  Conditional on an uncharged round, the battery state is
\(\ket0_W\), whose ergotropy is zero.  Therefore the average event-resolved
extractable energy is
\begin{equation}
\Delta\,\bbP[\mathrm{charged}]
=
\Delta\,p_{\rm succ}^{\mathcal G}(P).
\label{eq:event_resolved_extractable_energy}
\end{equation}

If, instead, the charge record is ignored and only the average battery state is
considered, then after many rounds the battery state is
\begin{equation}
\rho_W(p)
=
(1-p)\ket0\!\bra0
+
p\ket1\!\bra1,
\qquad
p=p_{\rm succ}^{\mathcal G}(P).
\label{eq:average_battery_state}
\end{equation}
Its mean stored energy is
\begin{equation}
E_W(p)=\Tr[H_W\rho_W(p)]=\Delta p.
\label{eq:mean_stored_energy}
\end{equation}
However, its unconditional ergotropy is not generally equal to \(\Delta p\).
For a two-level Hamiltonian, the passive state is obtained by putting the larger
eigenvalue of the density matrix on the ground state.  Hence
\begin{equation}
\mathcal W(\rho_W(p))
=
\Delta\max\{0,2p-1\}.
\label{eq:two_level_ergotropy}
\end{equation}
If \(p\le1/2\), the averaged battery state is passive.  If \(p>1/2\), it is
population-inverted and a unitary swap of the two levels extracts
\(\Delta(2p-1)\).

Thus the witness should be read as a statement about stored battery energy, or
about event-resolved extractable energy when the charge record is retained.  It
is not a statement that the ergotropy of the averaged battery state is always
\(\Delta p\).  This is the standard distinction between passivity and maximal
unitary work extraction \cite{PuszWoronowicz1978,AllahverdyanBalianNieuwenhuizen2004}.

The same terminology also clarifies the role of temperature.  For a reservoir at
inverse temperature \(\beta\), the thermal state of the isolated battery is
\begin{equation}
\gamma_W
=
\frac{\ket0\!\bra0+e^{-\beta\Delta}\ket1\!\bra1}
{1+e^{-\beta\Delta}}.
\label{eq:battery_gibbs_state_passivity}
\end{equation}
For positive temperature this state is passive.  The ideal battery-routing step
does not thermalize the battery; it is a closed energy-preserving excitation
transfer.  Temperature enters only when one models relaxation, thermal
preparation, or erasure/reset operations.

\subsection{Coherent charger--battery coupling}

The excitation transfer used in the battery module can be realized by a coherent
resonant coupling between the fuel and the battery.  On the \(F\otimes W\)
system consider
\begin{equation}
H_{\rm ex}
=
g\left(
\ket{01}\!\bra{10}_{FW}
+
\ket{10}\!\bra{01}_{FW}
\right),
\label{eq:exchange_hamiltonian}
\end{equation}
where \(g>0\).  This Hamiltonian acts only in the one-excitation subspace
\[
{\rm span}\{\ket{10}_{FW},\ket{01}_{FW}\}.
\]
Since both basis states in this subspace have total energy \(\Delta\), one has
\begin{equation}
[H_{\rm ex},H_F+H_W]=0.
\label{eq:exchange_energy_conservation}
\end{equation}
Thus the coupling is energy preserving.

Starting from an excited fuel and an empty battery,
\begin{equation}
e^{-itH_{\rm ex}}\ket{10}_{FW}
=
\cos(gt)\ket{10}_{FW}
-
i\sin(gt)\ket{01}_{FW}.
\label{eq:rabi_transfer}
\end{equation}
At the interaction time
\begin{equation}
t_*=\frac{\pi}{2g},
\label{eq:swap_time}
\end{equation}
this gives
\begin{equation}
e^{-it_*H_{\rm ex}}\ket{10}_{FW}
=
-i\ket{01}_{FW}.
\label{eq:swap_up_to_phase}
\end{equation}
Hence the excitation is transferred from the fuel to the battery, up to a phase.
For the classical work statistics considered here this phase has no effect.  If
one wants a coherent implementation on superpositions of different control
histories, an additional energy-preserving phase correction may be included.

The equality-controlled version is obtained by switching on the exchange only
on the branch \(X=G\):
\begin{equation}
H_{\rm ctrl}
=
\sum_{x,g:\,x=g}
\ket{x}\!\bra{x}_X
\otimes
\ket{g}\!\bra{g}_G
\otimes
H_{\rm ex}.
\label{eq:controlled_exchange_hamiltonian}
\end{equation}
The logical registers are degenerate, and \(H_{\rm ex}\) commutes with
\(H_F+H_W\).  Therefore
\begin{equation}
[H_{\rm ctrl},H_X+H_G+H_F+H_W]=0.
\label{eq:controlled_exchange_energy_conservation}
\end{equation}
Evolution under \(H_{\rm ctrl}\) for time \(t_*\) implements the desired
equality-controlled excitation transfer, again up to a phase irrelevant for the
classical battery population.

For completeness, the ideal controlled-SWAP unitary \(U_{\rm bat}\) can also be
written as a finite-time Hamiltonian evolution.  Let \(S_{FW}\) denote the SWAP
on \(F\otimes W\).  Since \(S_{FW}\) is Hermitian and unitary, define
\begin{equation}
K_{\rm swap}
=
\frac{\pi}{2\tau}\left(\Id_{FW}-S_{FW}\right),
\label{eq:Kswap}
\end{equation}
where \(\tau>0\) is a chosen gate time.  On the \(+1\) eigenspace of \(S_{FW}\),
\(K_{\rm swap}\) has eigenvalue \(0\), while on the \(-1\) eigenspace it has
eigenvalue \(\pi/\tau\).  Thus
\begin{equation}
e^{-iK_{\rm swap}\tau}=S_{FW}.
\label{eq:swap_generator}
\end{equation}
The corresponding controlled generator is
\begin{equation}
H_{\rm bat}
=
\sum_{x,g:\,x=g}
\ket{x}\!\bra{x}_X
\otimes
\ket{g}\!\bra{g}_G
\otimes
K_{\rm swap},
\label{eq:Hbat_generator}
\end{equation}
and satisfies
\begin{equation}
U_{\rm bat}=e^{-iH_{\rm bat}\tau}.
\label{eq:Ubat_exponential}
\end{equation}
This formal expression answers how the ideal unitary may be generated.  The
resonant exchange Hamiltonian in Eq.~\eqref{eq:exchange_hamiltonian} is the more
direct charger--battery implementation in the one-excitation sector used by the
protocol.

The same expression also shows how timing errors enter.  If the exchange pulse
lasts
\[
t=t_*+\delta t,
\]
then the transfer probability is
\begin{equation}
p_{\rm tr}(t)
=
\sin^2(gt)
=
\cos^2(g\delta t).
\label{eq:timing_error_transfer}
\end{equation}
For small \(\delta t\),
\begin{equation}
p_{\rm tr}(t)
=
1-g^2\delta t^2+O(\delta t^4).
\label{eq:timing_error_small}
\end{equation}
Thus an imperfect pulse reduces the true-positive efficiency of the
energy-routing module.  In the readout-imperfection model used later, this
contributes to the calibrated parameter \(\eta_1\).

\subsection{Bath, temperature, and dissipative corrections}
\label{subsec:bath_temperature}

The ideal witness identity is derived for the closed routing step.  No
dissipative environment is needed to obtain
\[
\bbE[W_{\rm bat}]
=
\Delta p_{\rm succ}^{\mathcal G}(P).
\]
The role of a thermal bath is different: it enters when one models relaxation of
the fuel or battery, thermal preparation, or erasure of a persistent memory
record.  In particular, the Landauer term in the cyclic bookkeeping section is
the place where the bath temperature \(T\) enters the ideal thermodynamic
balance.

If the battery is weakly coupled to a reservoir during or after the charging
step, the ideal unitary dynamics should be replaced by an open-system evolution.
A simple phenomenological description is given by the standard GKSL form
\cite{GoriniKossakowskiSudarshan1976,Lindblad1976}
\begin{equation}
\dot\rho
=
-i[H_{\rm ctrl}(t),\rho]
+
\gamma_\downarrow \mathcal D[\sigma_-^W]\rho
+
\gamma_\uparrow \mathcal D[\sigma_+^W]\rho
+
\gamma_\phi \mathcal D[\sigma_z^W]\rho ,
\label{eq:lindblad_battery}
\end{equation}
where
\begin{equation}
\mathcal D[L]\rho
=
L\rho L^\dagger
-\frac12\{L^\dagger L,\rho\}.
\label{eq:lindblad_dissipator}
\end{equation}
Here \(\gamma_\downarrow\) describes relaxation of the battery excitation,
\(\gamma_\uparrow\) describes thermal excitation, and \(\gamma_\phi\) describes
dephasing.  For a thermal reservoir at inverse temperature \(\beta=1/(\kB T)\),
the rates satisfy the detailed-balance relation
\begin{equation}
\frac{\gamma_\uparrow}{\gamma_\downarrow}
=
e^{-\beta\Delta}.
\label{eq:detailed_balance}
\end{equation}
The corresponding stationary state of the isolated battery is the Gibbs state
\begin{equation}
\gamma_W
=
\frac{\ket0\!\bra0+e^{-\beta\Delta}\ket1\!\bra1}
{1+e^{-\beta\Delta}},
\label{eq:battery_gibbs_temperature}
\end{equation}
which is passive for positive temperature.

Such dissipative effects are not part of the ideal witness theorem.  They
change the experimentally observed charging probability by producing losses,
false charges, or relaxation before readout.  In the notation of
Sec.~\ref{sec:statistics_robustness}, they contribute to the calibrated
parameters \(\eta_1\) and \(\eta_0\).  Thus dissipation does not change the
logical form of the witness; it changes the calibration needed to infer the
underlying game success probability from the observed battery signal.

\section{XOR-game battery transducer}
\label{sec:xor_transducer}

We now give the main construction.  The starting point is a finite two-player
XOR game
\[
\mathcal{G}=(\mathcal{U},\mathcal{V},\pi,f),
\]
where \(f:\mathcal{U}\times\mathcal{V}\to\{0,1\}\).  The referee samples
\((U,V)\sim\pi\), Alice and Bob output bits \(A,B\), and the game is won when
\[
A\oplus B=f(U,V).
\]
For a behaviour \(P(a,b|u,v)\), the success probability is
\[
p_{\rm succ}^{\mathcal{G}}(P)
:=
\bbP[A\oplus B=f(U,V)].
\]

\subsection{Equality construction}

The trusted referee/controller also samples an independent uniform bit \(R\).
This bit is not supplied to the devices producing \(A\) and \(B\), and the
behaviour \(P(a,b|u,v)\) is assumed not to depend on \(R\).  We define
\[
X:=f(U,V)\oplus R,
\qquad
G:=A\oplus B\oplus R.
\]
The bit \(X\) is the target bit and \(G\) is the guess bit.

\begin{lemma}[Winning is equality]
\label{lem:winning_equality}
For every realization \((u,v,a,b,r)\),
\[
G=X
\quad\Longleftrightarrow\quad
a\oplus b=f(u,v).
\]
Consequently,
\[
\bbP[G=X]=p_{\rm succ}^{\mathcal{G}}(P).
\]
\end{lemma}

\begin{proof}
We compute
\[
G\oplus X
=
(a\oplus b\oplus r)\oplus(f(u,v)\oplus r)
=
a\oplus b\oplus f(u,v),
\]
because \(r\oplus r=0\).  Therefore \(G=X\) if and only if
\(G\oplus X=0\), which is equivalent to \(a\oplus b=f(u,v)\).
Averaging over the distribution of all variables gives
\[
\bbP[G=X]
=
\bbP[A\oplus B=f(U,V)]
=
p_{\rm succ}^{\mathcal{G}}(P).
\]
\end{proof}

The random pad \(R\) also has a useful thermodynamic role.  It removes local
bias from the logical bits.

\begin{lemma}[No local bias in the equality registers]
\label{lem:no_local_bias}
The target bit \(X\) and the guess bit \(G\) are both uniform:
\[
\bbP[X=0]=\bbP[X=1]=\frac12,
\qquad
\bbP[G=0]=\bbP[G=1]=\frac12.
\]
\end{lemma}

\begin{proof}
For fixed \(u,v\),
\[
X=f(u,v)\oplus R.
\]
Since \(R\) is uniform,
\[
\bbP[X=x|U=u,V=v]=\frac12
\]
for both \(x=0,1\).  Hence \(X\) is uniform.

Now define the error bit
\[
E:=G\oplus X.
\]
From the same cancellation as above,
\[
E=A\oplus B\oplus f(U,V).
\]
Thus \(E\) depends on the game transcript but not on \(R\).  The one-time pad
makes \(X\) independent of the transcript and hence independent of \(E\).  Since
\[
G=X\oplus E,
\]
we get, for \(g\in\{0,1\}\),
\[
\bbP[G=g]
=
\sum_{e=0}^1 \bbP[E=e]\bbP[X=g\oplus e]
=
\sum_{e=0}^1 \bbP[E=e]\frac12
=
\frac12.
\]
\end{proof}

\begin{remark}[Correlation rather than local free energy]
Because \(X\) and \(G\) are individually uniform and stored in degenerate
logical registers, they carry no local nonequilibrium free energy in this
idealized model.  The relevant information is their correlation:
\[
\bbP[G=X]=p_{\rm succ}^{\mathcal{G}}(P).
\]
The battery module below converts this correlation, not a local bias of either
bit, into a charging probability.
\end{remark}

\subsection{Battery module}

The target and guess bits are stored in degenerate logical registers, also
called \(X\) and \(G\).  The fuel and battery have Hamiltonians
\[
H_F=\Delta\ket{1}\!\bra{1}_F,
\qquad
H_W=\Delta\ket{1}\!\bra{1}_W.
\]
At the beginning of the round,
\[
F=\ket{1}_F,
\qquad
W=\ket{0}_W.
\]

Let \(\mathrm{SWAP}_{FW}\) be the two-qubit SWAP on fuel and battery:
\[
\mathrm{SWAP}_{FW}\ket{i}_F\ket{j}_W
=
\ket{j}_F\ket{i}_W.
\]
Since \(F\) and \(W\) have the same gap,
\[
[\mathrm{SWAP}_{FW},H_F+H_W]=0.
\]
Indeed, the only nontrivial action in the one-excitation sector is
\[
\ket{1}_F\ket{0}_W
\leftrightarrow
\ket{0}_F\ket{1}_W,
\]
and both states have total energy \(\Delta\).

Define the equality-controlled battery unitary
\[
U_{\rm bat}
:=
\sum_{x,g\in\{0,1\}}
\ket{x}\!\bra{x}_X
\otimes
\ket{g}\!\bra{g}_G
\otimes
V_{xg},
\]
where
\[
V_{xg}
:=
\begin{cases}
\mathrm{SWAP}_{FW}, & x=g,\\
\Id_{FW}, & x\ne g.
\end{cases}
\]
Since the branches are controlled on orthogonal projectors and each branch is
unitary, \(U_{\rm bat}\) is unitary.  Since the logical registers are
degenerate and each branch commutes with \(H_F+H_W\),
\[
[U_{\rm bat},H_X+H_G+H_F+H_W]=0.
\]

\begin{theorem}[XOR-game battery transduction]
\label{thm:xor_battery_transduction}
For any finite two-player XOR game \(\mathcal{G}\) and any behaviour \(P\), the
battery module satisfies
\[
W_{\rm bat}
=
\Delta\,\mathbf{1}\{G=X\}
\]
in every run.  Therefore
\begin{equation}
\boxed{
\bbE[W_{\rm bat}]
=
\Delta\,p_{\rm succ}^{\mathcal{G}}(P).
}
\label{eq:work_success_main}
\end{equation}
\end{theorem}

\begin{proof}
If \(X=G\), the SWAP branch is applied, and
\[
\ket{1}_F\ket{0}_W
\longmapsto
\ket{0}_F\ket{1}_W.
\]
The battery gains energy \(\Delta\).  If \(X\ne G\), the identity branch is
applied and the battery remains in \(\ket{0}_W\), so it gains no energy.  Hence
\[
W_{\rm bat}
=
\Delta\,\mathbf{1}\{G=X\}.
\]
Taking expectations and using Lemma~\ref{lem:winning_equality} gives
\[
\bbE[W_{\rm bat}]
=
\Delta\,\bbP[G=X]
=
\Delta\,p_{\rm succ}^{\mathcal{G}}(P).
\]
\end{proof}

\begin{remark}[Why we focus on XOR games]
The battery mechanism itself is more general than XOR games.  Any finite
classical binary predicate \(V(T)\in\{0,1\}\), computed from a transcript \(T\),
can be used to control the same energy-preserving SWAP, giving
\[
\bbE[W_{\rm bat}]=\Delta\,\bbP[V(T)=1].
\]
We focus on XOR games because their local, quantum, and nonsignalling values are
standard and give familiar post-quantum thresholds, including CHSH and chained
Bell inequalities.  The fully general binary-predicate statement is recorded in
Appendix~\ref{app:general_predicate}.
\end{remark}

\begin{corollary}[Resource-class ceilings]
\label{cor:resource_ceilings}
Let \(\mathsf{C}\) be any class of behaviours for the game \(\mathcal{G}\), and
define
\[
\omega_{\mathsf{C}}(\mathcal{G})
:=
\sup_{P\in\mathsf{C}}p_{\rm succ}^{\mathcal{G}}(P).
\]
Then the maximal mean battery charge achievable within \(\mathsf{C}\) is
\begin{equation}
\boxed{
W_{\mathsf{C}}^{\max}(\mathcal{G})
=
\Delta\,\omega_{\mathsf{C}}(\mathcal{G}).
}
\label{eq:class_ceiling_main}
\end{equation}
\end{corollary}

\begin{proof}
By Theorem~\ref{thm:xor_battery_transduction},
\[
\bbE[W_{\rm bat}]
=
\Delta\,p_{\rm succ}^{\mathcal{G}}(P)
\]
for every behaviour \(P\).  Taking the supremum over \(P\in\mathsf{C}\) gives
\[
W_{\mathsf{C}}^{\max}(\mathcal{G})
=
\sup_{P\in\mathsf{C}}\bbE[W_{\rm bat}]
=
\Delta\sup_{P\in\mathsf{C}}p_{\rm succ}^{\mathcal{G}}(P)
=
\Delta\,\omega_{\mathsf{C}}(\mathcal{G}).
\]
\end{proof}

Thus, under the trusted-module assumptions, if
\[
\bbE[W_{\rm bat}]
>
\Delta\,\omega_{\mathsf{Q}}(\mathcal{G}),
\]
then either the effective behaviour is not quantum-realizable, or one of the
trusted assumptions about the battery module has failed.  With independently
validated module assumptions, this becomes a witness of post-quantumness.

\section{CHSH and chained-game witnesses}
\label{sec:examples}

We now specialize the general theorem to two standard families of XOR games.
The CHSH game gives the clearest post-quantum separation.  The chained games
show that the construction is not tied to the four-setting CHSH algebra.

\subsection{CHSH}

For CHSH,
\[
\mathcal{U}=\mathcal{V}=\{0,1\},
\qquad
\pi(u,v)=\frac14,
\qquad
f(u,v)=uv.
\]
The winning condition is
\[
a\oplus b=uv.
\]
Define correlators
\[
E_{uv}
:=
\sum_{a,b}
(-1)^{a\oplus b}P(a,b|u,v),
\]
and the CHSH expression
\[
S(P):=
E_{00}+E_{01}+E_{10}-E_{11}.
\]
With this convention, local behaviours satisfy \(S(P)\le2\), quantum
behaviours satisfy Tsirelson's bound \(S(P)\le2\sqrt2\), and nonsignalling
behaviours satisfy \(S(P)\le4\), with the algebraic value attained by a PR box
\cite{CHSH1969,Tsirelson1980,PopescuRohrlich1994,BarrettPRA2005}.

For CHSH, the success probability is
\[
p_{\rm win}(P)
=
\frac12+\frac{S(P)}8.
\]
Indeed, the XOR-game bias is
\[
\frac14(E_{00}+E_{01}+E_{10}-E_{11})
=
\frac{S(P)}4,
\]
and \(p_{\rm win}=(1+\mathrm{bias})/2\).

Therefore Theorem~\ref{thm:xor_battery_transduction} gives
\begin{equation}
\boxed{
\frac{\bbE[W_{\rm bat}]}{\Delta}
=
\frac12+\frac{S(P)}8.
}
\label{eq:chsh_work_relation}
\end{equation}

The standard CHSH bounds become the following battery ceilings:
\[
W_{\mathsf{L}}^{\max}
=
\frac34\,\Delta,
\]
\[
W_{\mathsf{Q}}^{\max}
=
\Delta\left(\frac12+\frac{\sqrt2}{4}\right)
=
\Delta\cos^2\frac{\pi}{8},
\]
and
\[
W_{\mathsf{NS}}^{\max}
=
\Delta.
\]
Thus a trusted observation of
\[
\bbE[W_{\rm bat}]
>
\Delta\cos^2\frac{\pi}{8}
\]
is a post-quantum witness for the effective CHSH behaviour, subject to the
trusted-module assumptions.

A noisy PR-box interpolation gives a simple benchmark.  Let
\[
P_\epsilon
=
(1-\epsilon)P_{\rm PR}+\epsilon P_{\rm L},
\]
where \(S(P_{\rm PR})=4\) and \(S(P_{\rm L})=2\).  Then
\[
S(P_\epsilon)=4-2\epsilon,
\]
and hence
\[
\bbE[W_{\rm bat}]
=
\Delta\left(1-\frac{\epsilon}{4}\right).
\]
This remains above the quantum CHSH ceiling exactly when
\[
4-2\epsilon>2\sqrt2,
\]
or
\[
\epsilon<2-\sqrt2.
\]

Figure~\ref{fig:work_vs_chsh} shows the affine relation between the CHSH value
and the normalized mean battery charge.

\begin{figure}[t]
\centering
\includegraphics[width=0.85\linewidth]{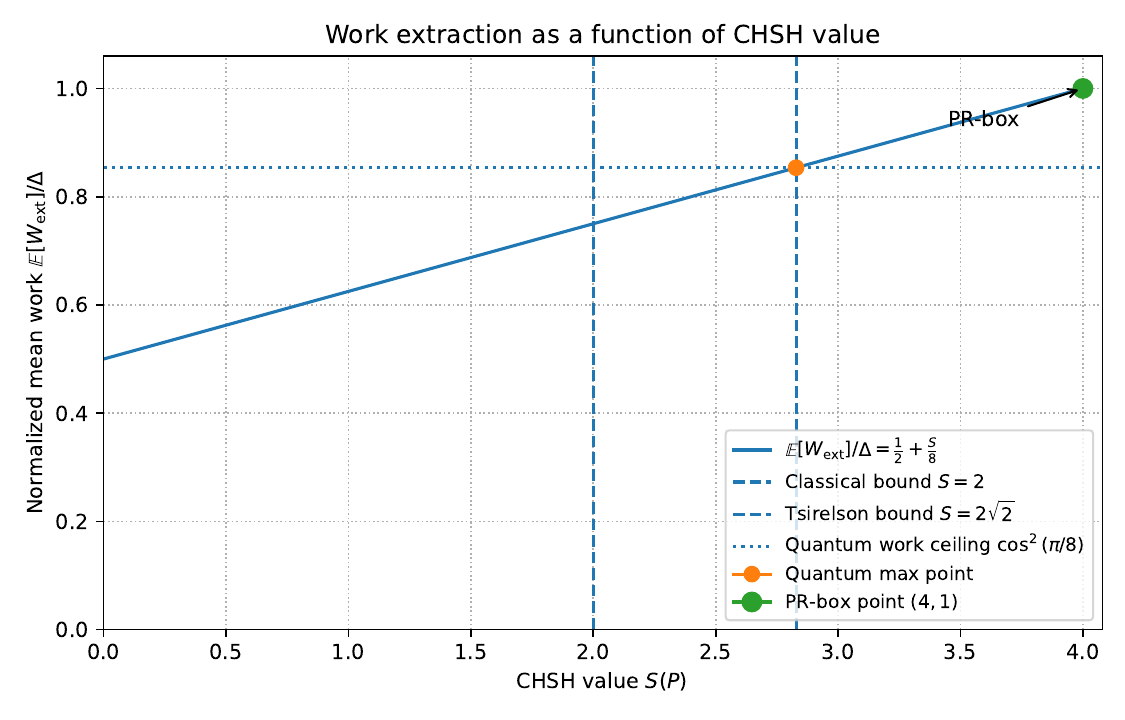}
\caption{
Mean battery charge for the CHSH game.  The normalized value satisfies
\(\bbE[W_{\rm bat}]/\Delta=\frac12+\frac{S(P)}8\).
The classical bound \(S=2\) gives \(3\Delta/4\).  Tsirelson's bound
\(S=2\sqrt2\) gives the quantum ceiling
\(\Delta\cos^2(\pi/8)\).  A PR box reaches the single-excitation cap
\(\Delta\).
}
\label{fig:work_vs_chsh}
\end{figure}

\subsection{Chained Bell games}

The chained Bell games form a standard family of XOR games generalizing CHSH
\cite{BraunsteinCaves1990,Wehner2006}.  Fix \(N\ge2\).  The question sets are
\[
\mathcal{U}=\mathcal{V}=\{0,1,\ldots,N-1\}.
\]
The referee samples uniformly from the \(2N\) input pairs
\[
(u,v)=(j,j),
\qquad
(u,v)=(j+1,j),
\]
where \(j=0,\ldots,N-1\) and addition is modulo \(N\).  All edges require equal
outputs except one wrap-around edge, which requires unequal outputs:
\[
f(j,j)=0,
\]
\[
f(j+1,j)=0
\quad
(j=0,\ldots,N-2),
\]
and
\[
f(0,N-1)=1.
\]

For this game, the standard values are
\[
\omega_{\mathsf{L}}(\mathcal{G}_N)
=
1-\frac{1}{2N},
\]
\[
\omega_{\mathsf{Q}}(\mathcal{G}_N)
=
\cos^2\left(\frac{\pi}{4N}\right),
\]
and
\[
\omega_{\mathsf{NS}}(\mathcal{G}_N)
=
1.
\]
The local value follows because deterministic assignments cannot satisfy all
\(2N\) parity constraints around the cycle, while they can satisfy \(2N-1\) of
them.  The quantum value is the chained Tsirelson value
\cite{BraunsteinCaves1990,Wehner2006}.  The nonsignalling value is one, because
a nonsignalling box can satisfy each allowed XOR constraint with uniformly
random local marginals.

By Corollary~\ref{cor:resource_ceilings}, the corresponding battery ceilings are
\[
W_{\mathsf{L}}^{\max}(\mathcal{G}_N)
=
\Delta\left(1-\frac{1}{2N}\right),
\]
\[
W_{\mathsf{Q}}^{\max}(\mathcal{G}_N)
=
\Delta\cos^2\left(\frac{\pi}{4N}\right),
\]
and
\[
W_{\mathsf{NS}}^{\max}(\mathcal{G}_N)
=
\Delta.
\]
Thus, under the trusted-module assumptions,
\[
\bbE[W_{\rm bat}]
>
\Delta\cos^2\left(\frac{\pi}{4N}\right)
\]
witnesses that the effective chained-game behaviour is not quantum-realizable.

For large \(N\), the quantum-to-nonsignalling battery gap is
\[
\Delta-
\Delta\cos^2\left(\frac{\pi}{4N}\right)
=
\Delta\sin^2\left(\frac{\pi}{4N}\right)
=
\Delta\left[
\frac{\pi^2}{16N^2}
+
O\left(\frac{1}{N^4}\right)
\right].
\]
Thus the gap becomes smaller as \(N\) grows.  Within this family, CHSH gives the
largest and cleanest quantum-to-PR separation.

\section{Reversible controller and cyclic bookkeeping}
\label{sec:controller_bookkeeping}

The battery transducer can be described in feed-forward form: after the game
outputs are produced, the equality condition \(G=X\) is computed and used to
control the fuel-battery SWAP.  We now explain how this feed-forward step can be
implemented reversibly on degenerate logical registers, and then discuss the
cyclic thermodynamic bookkeeping.

\subsection{Reversible-controller implementation}

Let \(D\) denote the classical data registers storing the transcript
\[
(u,v,a,b,r).
\]
We take these registers to be degenerate:
\[
H_D=0.
\]
Let \(M\) be a degenerate memory bit initialized in \(\ket{0}_M\), and let
\(A_{\rm anc}\) denote finitely many degenerate ancilla bits initialized in
\(\ket{0\cdots0}_{A_{\rm anc}}\).  Thus all logical controller registers have
zero Hamiltonian.

For a fixed game transcript, define the success bit
\[
Z(u,v,a,b)
:=
\mathbf{1}\{a\oplus b=f(u,v)\}.
\]
Equivalently,
\[
Z(u,v,a,b)
=
1\oplus a\oplus b\oplus f(u,v).
\]
Since the game is finite, \(f\) is a finite Boolean function.  Hence it can be
computed reversibly using standard reversible logic, with enough temporary
ancillas.  We may therefore choose a reversible circuit \(V_{\rm comp}\) such
that
\[
V_{\rm comp}
\ket{u,v,a,b,r}_D
\ket{0}_M
\ket{0\cdots0}_{A_{\rm anc}}
=
\ket{u,v,a,b,r}_D
\ket{Z(u,v,a,b)}_M
\ket{0\cdots0}_{A_{\rm anc}}.
\]
All temporary workspace is uncomputed at the end of \(V_{\rm comp}\).

The feedback SWAP is
\[
U_{\rm fb}
=
\ket{1}\!\bra{1}_M\otimes \mathrm{SWAP}_{FW}
+
\ket{0}\!\bra{0}_M\otimes \Id_{FW}.
\]
Since \(M\) is degenerate and \(\mathrm{SWAP}_{FW}\) commutes with
\(H_F+H_W\),
\[
[U_{\rm fb},H_M+H_F+H_W]=0.
\]
Define
\[
U_{\rm rev}
:=
V_{\rm comp}^{\dagger}U_{\rm fb}V_{\rm comp}.
\]
Then \(U_{\rm rev}\) is energy preserving with respect to the explicit
fuel-battery Hamiltonian and implements the same work statistics as the
feed-forward transducer.

Indeed, \(V_{\rm comp}\) writes the success bit \(Z\) into \(M\), \(U_{\rm fb}\)
moves the excitation from \(F\) to \(W\) if and only if \(Z=1\), and
\(V_{\rm comp}^{\dagger}\) restores \(M\) and the temporary ancillas to their
initial blank states.  Thus
\[
W_{\rm bat}
=
\Delta\,\mathbf{1}\{a\oplus b=f(u,v)\},
\]
and therefore
\[
\bbE[W_{\rm bat}]
=
\Delta\,p_{\rm succ}^{\mathcal{G}}(P).
\]

\begin{remark}[Reversible, not fully autonomous]
The construction above is a reversible-controller implementation.  We do not
claim here a fully autonomous clock-Hamiltonian model of the controller.  The
logical gates are trusted control operations on degenerate registers.  Their
role is to show that the feed-forward computation need not change the explicit
fuel-battery energy bookkeeping.
\end{remark}

\begin{remark}[Hamiltonian implementation]
A Hamiltonian implementation of the controlled excitation transfer, including
the phase appearing in the resonant exchange pulse, is discussed in
Sec.~\ref{sec:battery_passivity_charging}.  In the reversible-controller
description used here the control registers are classical computational-basis
registers, so this phase does not affect the battery statistics.
\end{remark}

\subsection{Boundary of the thermodynamic cycle}

The reversible-controller construction preserves the transcript registers
\(D=(u,v,a,b,r)\).  In the cyclic bookkeeping below, these registers are treated
as external game data supplied by the referee/game interface.  The thermodynamic
cycle explicitly includes only the local fuel, battery, controller memory when
it is persistently stored, and their reset or restoration operations.

If a concrete implementation instead reuses the same physical registers for the
full transcript, then their erasure or overwrite cost must also be included.
That additional cost is at least the Landauer cost of the stored transcript and
is no smaller than the cost of erasing the compressed success/failure bit.

\subsection{Fuel restoration}

Let
\[
p:=p_{\rm succ}^{\mathcal{G}}(P).
\]
The mean battery charge is
\[
\bbE[W_{\rm bat}]=\Delta p.
\]

In a successful round, the fuel-battery state changes as
\[
\ket{1}_F\ket{0}_W
\longmapsto
\ket{0}_F\ket{1}_W.
\]
The fuel excitation has been consumed.  To reuse the same fuel system in the
initial state \(\ket{1}_F\), one must restore the excitation, which costs at
least \(\Delta\) in the ideal energy-eigenstate model.

In a failed round, the state remains
\[
\ket{1}_F\ket{0}_W.
\]
The fuel excitation is still present and no fuel re-excitation is needed.  Hence
the minimal average fuel restoration cost is success-weighted:
\[
\bbE[W_{\rm fuel}]\ge \Delta p.
\]
This is the relevant cyclic accounting when unused fuel excitations are retained
and reused.

\subsection{Reversible-controller cycle}

In the reversible-controller implementation, the success bit is computed, used,
and uncomputed.  No persistent success/failure memory remains.  Thus there is
no Landauer erasure cost for that bit.

The fuel-battery contribution to the cycle satisfies
\[
\bbE[W_{\rm bat}]-\bbE[W_{\rm fuel}]
\le
\Delta p-\Delta p
=
0.
\]
In the ideal reversible limit this balance is saturated.  Thus the reversible
implementation gives no positive net work once the restoration of consumed fuel
excitations is included.

\subsection{Measured-memory implementation}

A different implementation may measure or persistently store the success/failure
bit
\[
Z=\mathbf{1}\{a\oplus b=f(u,v)\}.
\]
Then
\[
\bbP[Z=1]=p,
\qquad
\bbP[Z=0]=1-p,
\]
and
\begin{equation}
H(Z)=h_2(p),
\qquad
h_2(p):=-p\log_2 p-(1-p)\log_2(1-p).
\label{eq:binary_entropy_success_memory}
\end{equation}
If this memory is blindly reset before the next run, Landauer's principle gives
\cite{Landauer1961,Bennett1982,ReebWolf2014,Goold2016}
\[
Q_{\rm reset}
\ge
\kB T\ln2\,h_2(p).
\]

Therefore the measured-memory cycle obeys
\[
\bbE[W_{\rm net}]
\le
\bbE[W_{\rm bat}]
-
\bbE[W_{\rm fuel}]
-
\bbE[Q_{\rm reset}],
\]
and hence
\begin{equation}
\bbE[W_{\rm net}]
\le
-\kB T\ln2\,h_2(p)
\le0.
\label{eq:measured_memory_no_positive_work}
\end{equation}

For a perfect nonsignalling strategy, \(p=1\), the success/failure memory is
deterministic and \(h_2(1)=0\).  Even then, the battery charge is exactly
balanced by the fuel restoration cost.  Thus neither quantum nor post-quantum
correlations generate positive cyclic work in this model.

\subsection{Summary of the thermodynamic role}

The battery transducer is an energy-routing witness.  The supplied excitation is
the energetic resource.  The correlation resource determines the probability
with which that excitation reaches the battery.  The cyclic bookkeeping confirms
that the protocol does not convert Bell nonlocality or post-quantumness into a
thermodynamic fuel.

\section{Finite statistics and readout imperfections}
\label{sec:statistics_robustness}

In the ideal transducer, each round produces a binary work value
\[
W_i\in\{0,\Delta\}.
\]
Thus the game success probability can be estimated directly from battery data.
This section gives simple finite-statistics and readout-robustness statements.

\subsection{Finite statistics from work data}

Define the observed work bit
\[
Z_i:=\frac{W_i}{\Delta}\in\{0,1\},
\]
and let
\[
\hat p_n
:=
\frac1n\sum_{i=1}^n Z_i.
\]
In the ideal i.i.d. setting,
\[
\bbE[Z_i]=p_{\rm succ}^{\mathcal{G}}(P).
\]
Hoeffding's inequality \cite{Hoeffding1963} gives
\[
\bbP\left[
p_{\rm succ}^{\mathcal{G}}(P)
<
\hat p_n-\varepsilon
\right]
\le
e^{-2n\varepsilon^2}.
\]
Thus, for error probability \(\alpha\), define
\[
\varepsilon_n(\alpha)
:=
\sqrt{\frac{1}{2n}\ln\frac1\alpha}.
\]
Then, except with probability at most \(\alpha\),
\[
p_{\rm succ}^{\mathcal{G}}(P)
\ge
p_L
:=
\hat p_n-\varepsilon_n(\alpha).
\]

A finite-data nonlocality certificate is obtained if
\[
p_L>\omega_{\mathsf{L}}(\mathcal{G}),
\]
and a finite-data post-quantumness certificate is obtained if
\[
p_L>\omega_{\mathsf{Q}}(\mathcal{G}).
\]
Equivalently, in work units,
\[
\frac{\hat W_n}{\Delta}
-
\varepsilon_n(\alpha)
>
\omega_{\mathsf{Q}}(\mathcal{G}),
\]
where
\[
\hat W_n:=\frac1n\sum_{i=1}^n W_i.
\]

For CHSH, this becomes
\[
\hat p_n-\varepsilon_n(\alpha)
>
\cos^2\frac{\pi}{8}.
\]
Equivalently, using
\[
S=8\left(p_{\rm win}-\frac12\right),
\]
a lower confidence bound on \(S\) is
\[
S_L
=
8\left(
\hat p_n-\varepsilon_n(\alpha)-\frac12
\right),
\]
and post-quantumness is certified if
\[
S_L>2\sqrt2.
\]

\subsection{Martingale version}

The i.i.d. assumption can be weakened.  Let \(\mathcal{F}_{i-1}\) be the
history before round \(i\), and define
\[
p_i:=\bbE[Z_i|\mathcal{F}_{i-1}].
\]
Let
\[
\bar p_n:=\frac1n\sum_{i=1}^n p_i.
\]
The variables
\[
D_i:=Z_i-p_i
\]
are martingale differences with
\[
\bbE[D_i|\mathcal{F}_{i-1}]=0,
\qquad
|D_i|\le1.
\]
Azuma--Hoeffding gives \cite{Azuma1967,Hoeffding1963}
\[
\bbP\left[
\bar p_n
<
\hat p_n-\varepsilon
\right]
\le
e^{-2n\varepsilon^2}.
\]
Thus the same lower confidence bound applies to the time-averaged success
probability:
\[
\bar p_n
\ge
\hat p_n-\sqrt{\frac{1}{2n}\ln\frac1\alpha}
\]
except with probability at most \(\alpha\).

This version is useful when the effective behaviour may drift between rounds.
The conclusion then concerns the average success probability over the tested
rounds.

\subsection{Imperfect battery readout}

We now include a simple readout-error model.  Let
\[
\eta_1
:=
\bbP[\text{battery reports charged}\mid \mathrm{win}]
\]
be the true-positive probability, and let
\[
\eta_0
:=
\bbP[\text{battery reports charged}\mid \mathrm{fail}]
\]
be the false-positive probability.  Assume
\[
\eta_1>\eta_0.
\]
If
\[
p=p_{\rm succ}^{\mathcal{G}}(P),
\]
then the observed charging probability is
\[
p_{\rm obs}
=
\eta_1 p+\eta_0(1-p)
=
\eta_0+(\eta_1-\eta_0)p.
\]
Therefore
\[
p
=
\frac{p_{\rm obs}-\eta_0}{\eta_1-\eta_0}.
\]

If the calibration parameters are known, a lower confidence bound
\(p_{{\rm obs},L}\) gives
\[
p_{\rm succ}^{\mathcal{G}}(P)
\ge
\frac{p_{{\rm obs},L}-\eta_0}{\eta_1-\eta_0}.
\]
Thus the corrected post-quantumness condition is
\[
\frac{p_{{\rm obs},L}-\eta_0}{\eta_1-\eta_0}
>
\omega_{\mathsf{Q}}(\mathcal{G}).
\]

\subsection{Imperfect random pad}
\label{subsec:imperfect_pad}

The ideal construction assumes that the pad bit \(R\) is uniform, private, and
independent of the devices producing \(A,B\).  Uniformity is used to remove
local bias from the equality registers.  It is not needed for the equality
identity itself.

Suppose instead that \(R\) is still private and independent, but slightly biased:
\begin{equation}
\bbP[R=1]=\frac12+\delta,
\qquad
|\delta|\le\frac12.
\label{eq:biased_R}
\end{equation}
Then the algebraic cancellation
\[
G\oplus X
=
A\oplus B\oplus f(U,V)
\]
still holds.  Hence
\begin{equation}
G=X
\quad\Longleftrightarrow\quad
A\oplus B=f(U,V),
\label{eq:equality_survives_biased_R}
\end{equation}
and the battery identity
\begin{equation}
\bbE[W_{\rm bat}]
=
\Delta p_{\rm succ}^{\mathcal G}(P)
\label{eq:work_identity_biased_R}
\end{equation}
is unchanged.

What changes is the local bias of the logical registers.  Let
\[
q_f:=\bbP[f(U,V)=1].
\]
Then
\begin{align}
\bbP[X=1]
&=
\bbP[f(U,V)\oplus R=1]\\
&=
q_f\left(\frac12-\delta\right)
+
(1-q_f)\left(\frac12+\delta\right)\\
&=
\frac12+\delta(1-2q_f).
\label{eq:X_bias}
\end{align}
Therefore
\begin{equation}
\left|\bbP[X=1]-\frac12\right|
\le
|\delta|.
\label{eq:X_bias_bound}
\end{equation}
Similarly, if
\[
q_{AB}:=\bbP[A\oplus B=1],
\]
then
\begin{equation}
\bbP[G=1]
=
\frac12+\delta(1-2q_{AB}),
\label{eq:G_bias}
\end{equation}
and hence
\begin{equation}
\left|\bbP[G=1]-\frac12\right|
\le
|\delta|.
\label{eq:G_bias_bound}
\end{equation}

Thus imperfect randomness of \(R\) does not spoil the work--success-probability
identity, provided \(R\) remains private and independent of the devices.  It
does, however, weaken the statement that the individual logical registers have
no local bias.  If the logical registers are exactly degenerate, this bias does
not change their internal energy.  If one assigns an additional information
bookkeeping cost to biased logical records, that cost must be accounted for
separately.

If \(R\) is available to the devices producing \(A,B\), the situation is
different: the behaviour may effectively depend on \(r\), and the usual
resource-class values \(\omega_{\mathsf L}\), \(\omega_{\mathsf Q}\), and
\(\omega_{\mathsf{NS}}\) for the original game no longer apply without
modification.  For this reason the pad is assumed to be generated inside the
trusted referee/controller module and not supplied to the correlation resource.

\subsection{Conservative calibration bounds}

If the calibration parameters are not known exactly, the conservative lower
bound must be chosen carefully.  Since
\[
p=
\frac{p_{\rm obs}-\eta_0}{\eta_1-\eta_0},
\]
the inferred value of \(p\) decreases when \(\eta_0\) increases, and it also
decreases when \(\eta_1\) increases, provided the numerator is positive.

Therefore, if calibration gives
\[
\eta_0\le \eta_0^+,
\qquad
\eta_1\le \eta_1^+,
\]
with
\[
\eta_1^+>\eta_0^+,
\]
then a conservative lower bound is
\[
p_{\rm succ}^{\mathcal{G}}(P)
\ge
\frac{p_{{\rm obs},L}-\eta_0^+}{\eta_1^+-\eta_0^+},
\]
whenever the numerator is positive.  If no useful upper bound on \(\eta_1\) is
available, one may use the trivial upper bound \(\eta_1\le1\), giving
\[
p_{\rm succ}^{\mathcal{G}}(P)
\ge
\frac{p_{{\rm obs},L}-\eta_0^+}{1-\eta_0^+}.
\]
The right-hand side should be truncated to the interval \([0,1]\).

For CHSH, the corrected post-quantumness condition is
\[
\frac{p_{{\rm obs},L}-\eta_0^+}{\eta_1^+-\eta_0^+}
>
\cos^2\frac{\pi}{8}.
\]

\subsection{Symmetric work-bit flips}

As a simple benchmark, suppose the observed work bit is obtained from the ideal
work bit by a symmetric flip with probability \(\varepsilon\).  Then
\[
\eta_1=1-\varepsilon,
\qquad
\eta_0=\varepsilon.
\]
For an ideal PR-box behaviour, \(p=1\), so
\[
p_{\rm obs}=1-\varepsilon.
\]
To remain above the CHSH quantum ceiling without statistical uncertainty, one
needs
\[
1-\varepsilon
>
\cos^2\frac{\pi}{8}.
\]
Thus
\[
\varepsilon
<
1-\cos^2\frac{\pi}{8}
=
\sin^2\frac{\pi}{8}
\approx0.146447.
\]
This is the maximum symmetric work-bit error rate for which an ideal PR-box
signal remains above the quantum CHSH battery ceiling.

\section{Discussion}
\label{sec:discussion}

We have introduced a trusted battery witness of post-quantum Bell correlations.
The construction takes the success event of a finite XOR game and writes it into
an explicit work-storage degree of freedom.  In each round, a supplied excitation
is transferred from a fuel qubit to a battery qubit exactly when the game is won.
The resulting mean battery charge is
\[
\bbE[W_{\rm bat}]
=
\Delta\,p_{\rm succ}^{\mathcal{G}}(P).
\]
Thus the local, quantum, and nonsignalling values of the game become local,
quantum, and nonsignalling battery-charge ceilings.

For CHSH, this gives
\[
\frac{\bbE[W_{\rm bat}]}{\Delta}
=
\frac12+\frac{S(P)}8.
\]
The classical, quantum, and nonsignalling thresholds are respectively
\[
\frac34,
\qquad
\cos^2\frac{\pi}{8},
\qquad
1.
\]
A mean battery charge above the Tsirelson-calibrated threshold therefore
witnesses post-quantumness of the effective behaviour, provided the
thermodynamic module is independently trusted.  The chained Bell games give a
second family of examples, where the same transduction theorem converts the
known game values into corresponding battery thresholds.

The construction is intentionally modest in its thermodynamic claims.  It is not
a heat engine and it does not convert Bell correlations into net work.  The
energetic resource is the supplied fuel excitation.  The correlation resource
only determines the probability that this excitation is routed into the battery.
When the protocol is made cyclic, successful rounds require restoration of the
fuel excitation.  If a persistent success/failure memory is used, Landauer's
principle adds an erasure cost.  With these costs included, the full-cycle net
work is non-positive.

This distinction is important for interpreting the witness.  The battery charge
is not the thermodynamic value of the Bell correlation as a fuel.  Rather, it is
a calibrated energetic representation of the Bell-game success probability.  The
role of thermodynamics is to ensure that the readout is implemented without
hiding energy in the controller: the battery is charged only by an
energy-preserving transfer of a supplied excitation.

As discussed in Sec.~\ref{sec:battery_passivity_charging}, this stored energy is
equal to event-resolved extractable energy when the charge record is retained.
If the charge record is ignored, the ergotropy of the averaged battery state is
instead \(\Delta\max\{0,2p-1\}\), with
\(p=p_{\rm succ}^{\mathcal G}(P)\).

The one-time-pad bit \(R\) also has a useful thermodynamic interpretation.  It
makes the target and guess bits individually unbiased.  In the degenerate
logical-register model, neither bit carries local nonequilibrium free energy.
The game performance is instead encoded in the correlation between them, namely
in the probability that \(G=X\).  The battery transducer converts precisely this
correlation into a charging probability.

The witness is not device-independent.  It assumes trusted Hamiltonians, trusted
classical wiring, a calibrated battery gap, a correctly implemented
energy-preserving controlled SWAP, and a calibrated battery readout.  If these
assumptions are not independently validated, a work value above the quantum
ceiling could indicate a failure of the trusted module rather than genuine
post-quantumness. Thus the operational statement is conditional: under the
trusted-module assumptions, a battery charge above the quantum game value
certifies that the effective behaviour is not quantum-realizable.

The model is idealized, but the idealizations are explicit.  Imperfect
exchange pulses, battery relaxation, thermal excitation, dephasing, biased pad
randomness, and readout errors do not change the algebraic structure of the
witness; they change the calibration relating the observed battery signal to
the underlying game success probability.  For this reason the witness should be
used as a calibrated thermodynamic module.  Its conclusion is conditional: a
work value above the quantum threshold indicates either post-quantum effective
correlations or failure of one of the trusted calibrations.

The present construction is complementary to information-engine approaches. In
a Szilard-engine setting, correlations or side information can change the work
available to a feedback controller through mutual information.  A complementary
CHSH-based side-information construction was considered in
Ref.~\cite{cwiklinski2026_Szilard}. The difference is that the present paper
does not assign work value to side information.  Instead, it routes an explicit
fuel excitation into a battery according to the Bell-game winning event. It is also distinct from heat-based witnesses of quantum properties
\cite{deOliveiraJuniorBraskLipkaBartosik2025}: here the signal is not heat
generated in contact with a reservoir, but the charge of an explicit battery
receiving an excitation supplied at the beginning of the round.

Several extensions are natural.  One can study other XOR games for which the
gap between the quantum and nonsignalling values is more robust to noise or
finite statistics.  One can replace the ideal two-level battery by more realistic
work-storage systems with finite resolution, finite-time control, or small
Hamiltonian mismatches. One can also refine the statistical analysis to include
loss, input-dependent detection efficiency, or composable confidence bounds in
the presence of memory effects.  Finally, the same basic idea can be applied to
more general finite games with binary winning predicates: an energy-preserving
controlled SWAP can route a supplied excitation whenever the predicate is
satisfied.  XOR games were chosen here because they give a clean connection to
CHSH, chained Bell inequalities, and standard resource-class game values.

The general binary-predicate formulation also clarifies the scope of the result:
the thermodynamic module realizes a calibrated energy readout of a classical
predicate, while the physics enters through which predicate and which resource
class determine the relevant threshold.

The main message is therefore simple.  Bell-game performance can be written
directly into an explicit battery degree of freedom without making the controller
an unaccounted work source.  The result is a trusted thermodynamic witness: not
a new principle limiting correlations, and not a work-extraction engine, but a
battery-explicit way of representing local, quantum, and post-quantum
correlation strengths.

\section{Acknowledgements}
We thank Marcin Pawłowski, Sumit Rout and Anubhav Chaturvedi for discussions. This work is supported by the IRAP/MAB programme, project no. FENG.02.01IP.05-0006/23, financed by the MAB FENG program 2021-2027, Priority FENG.02, Measure FENG.02.01., with the support of the FNP (Foundation for Polish Science).

\appendix

\section{Detailed proof of the battery transduction theorem}
\label{app:detailed_transduction}

This appendix gives the details behind the energy-preserving transduction
statement used in Sec.~\ref{sec:xor_transducer}.  The main text keeps only the
short proof.

\subsection{Independence induced by the one-time pad}

Recall that the referee samples an independent uniform bit \(R\), and defines
\[
X=f(U,V)\oplus R,
\qquad
G=A\oplus B\oplus R.
\]
Define also the error bit
\[
E:=G\oplus X.
\]
Then
\[
E=A\oplus B\oplus f(U,V).
\]
Thus \(E\) depends on \((U,V,A,B)\), but not on \(R\).

\begin{lemma}[Independence of \(X\) and \(E\)]
\label{lem:appendix_XE_independent}
The target bit \(X\) is independent of the error bit \(E\).
\end{lemma}

\begin{proof}
Let \(x,e\in\{0,1\}\).  By the law of total probability,
\begin{align}
\bbP[X=x,E=e]
&=
\sum_{u,v}\pi(u,v)
\sum_{a,b}P(a,b|u,v)
\sum_{r=0}^{1}
\bbP[R=r]\,
\mathbf{1}\{x=f(u,v)\oplus r\}
\mathbf{1}\{e=a\oplus b\oplus f(u,v)\}.
\end{align}
For fixed \(u,v,x\), there is exactly one value of \(r\) satisfying
\[
x=f(u,v)\oplus r,
\]
namely
\[
r=x\oplus f(u,v).
\]
Since \(R\) is uniform,
\[
\bbP[R=r]=\frac12.
\]
Therefore the sum over \(r\) gives a factor \(1/2\), and
\begin{align}
\bbP[X=x,E=e]
&=
\frac12
\sum_{u,v}\pi(u,v)
\sum_{a,b}P(a,b|u,v)
\mathbf{1}\{e=a\oplus b\oplus f(u,v)\} \\
&=
\frac12\,\bbP[E=e].
\end{align}
Since \(X\) is uniform,
\[
\bbP[X=x]=\frac12.
\]
Thus
\[
\bbP[X=x,E=e]=\bbP[X=x]\bbP[E=e].
\]
Hence \(X\) and \(E\) are independent.
\end{proof}

\begin{corollary}[Binary symmetric channel]
\label{cor:appendix_bsc}
The relation between \(X\) and \(G\) is a binary symmetric channel:
\[
G=X\oplus E,
\]
where \(E\) is independent of \(X\).  Its crossover probability is
\[
q=\bbP[E=1]=1-p_{\rm succ}^{\mathcal G}(P).
\]
\end{corollary}

\begin{proof}
By definition,
\[
G=X\oplus E.
\]
By Lemma~\ref{lem:appendix_XE_independent}, \(E\) is independent of \(X\).
Moreover,
\[
E=0
\quad\Longleftrightarrow\quad
G=X
\quad\Longleftrightarrow\quad
A\oplus B=f(U,V).
\]
Therefore
\[
\bbP[E=0]=p_{\rm succ}^{\mathcal G}(P),
\]
and hence
\[
q=\bbP[E=1]=1-p_{\rm succ}^{\mathcal G}(P).
\]
\end{proof}

\subsection{Energy preservation of the equal-gap SWAP}

The fuel and battery Hamiltonians are
\[
H_F=\Delta\ket{1}\!\bra{1}_F,
\qquad
H_W=\Delta\ket{1}\!\bra{1}_W.
\]
The computational basis vectors of \(F\otimes W\) have energies
\[
E_{00}=0,
\qquad
E_{10}=\Delta,
\qquad
E_{01}=\Delta,
\qquad
E_{11}=2\Delta.
\]
The SWAP unitary satisfies
\[
\mathrm{SWAP}_{FW}\ket{00}=\ket{00},
\]
\[
\mathrm{SWAP}_{FW}\ket{10}=\ket{01},
\]
\[
\mathrm{SWAP}_{FW}\ket{01}=\ket{10},
\]
and
\[
\mathrm{SWAP}_{FW}\ket{11}=\ket{11}.
\]
It leaves the zero- and two-excitation sectors fixed, and it only exchanges the
two states in the one-excitation sector.  Since the one-excitation states have
the same total energy, the SWAP preserves every eigenspace of \(H_F+H_W\).
Therefore
\[
[\mathrm{SWAP}_{FW},H_F+H_W]=0.
\]

\subsection{Unitarity of the equality-controlled battery operation}

The equality-controlled battery unitary is
\[
U_{\rm bat}
=
\sum_{x,g\in\{0,1\}}
\ket{x}\!\bra{x}_X
\otimes
\ket{g}\!\bra{g}_G
\otimes
V_{xg},
\]
where
\[
V_{xg}
=
\begin{cases}
\mathrm{SWAP}_{FW}, & x=g,\\
\Id_{FW}, & x\ne g.
\end{cases}
\]
Let
\[
\Pi_{xg}
:=
\ket{x}\!\bra{x}_X\otimes\ket{g}\!\bra{g}_G.
\]
The projectors \(\Pi_{xg}\) are mutually orthogonal and resolve the identity:
\[
\Pi_{xg}\Pi_{x'g'}
=
\delta_{x,x'}\delta_{g,g'}\Pi_{xg},
\qquad
\sum_{x,g}\Pi_{xg}=\Id_{XG}.
\]
Each \(V_{xg}\) is unitary.  Therefore
\begin{align}
U_{\rm bat}^{\dagger}U_{\rm bat}
&=
\left(\sum_{x,g}\Pi_{xg}\otimes V_{xg}^{\dagger}\right)
\left(\sum_{x',g'}\Pi_{x'g'}\otimes V_{x'g'}\right)\\
&=
\sum_{x,g,x',g'}
\Pi_{xg}\Pi_{x'g'}
\otimes
V_{xg}^{\dagger}V_{x'g'}\\
&=
\sum_{x,g}
\Pi_{xg}\otimes V_{xg}^{\dagger}V_{xg}\\
&=
\sum_{x,g}
\Pi_{xg}\otimes \Id_{FW}\\
&=
\Id_{XGFW}.
\end{align}
The same calculation gives
\[
U_{\rm bat}U_{\rm bat}^{\dagger}=\Id_{XGFW}.
\]
Thus \(U_{\rm bat}\) is unitary.

\subsection{Energy preservation of the equality-controlled operation}

The logical registers are degenerate:
\[
H_X=H_G=0.
\]
Hence
\[
H_{\rm tot}=H_X+H_G+H_F+H_W=H_F+H_W.
\]
For each branch, \(V_{xg}\) is either the identity or \(\mathrm{SWAP}_{FW}\).
Both commute with \(H_F+H_W\).  Therefore every block
\[
\Pi_{xg}\otimes V_{xg}
\]
commutes with \(H_{\rm tot}\), and so does their sum:
\[
[U_{\rm bat},H_{\rm tot}]=0.
\]

\section{General binary-predicate battery routing}
\label{app:general_predicate}

The main text focuses on XOR games because they connect directly to CHSH,
chained Bell inequalities, and standard local, quantum, and nonsignalling game
values.  However, the energy-routing step is more general.

Let \(T\) be any finite classical transcript, distributed according to some
probability distribution, and let
\[
V:T\to\{0,1\}
\]
be a binary predicate.  The event \(V(T)=1\) may represent winning a game,
satisfying a Bell predicate, passing a test, or any other classical condition.

\begin{proposition}[Battery routing for a binary predicate]
\label{prop:binary_predicate_routing}
Let \(V(T)\in\{0,1\}\) be reversibly computable into a degenerate memory bit.
Then there is an energy-preserving fuel-battery unitary such that
\[
W_{\rm bat}=\Delta V(T)
\]
in every run.  Consequently,
\[
\bbE[W_{\rm bat}]=\Delta\,\bbP[V(T)=1].
\]
\end{proposition}

\begin{proof}
Let \(M\) be a degenerate memory bit initialized in \(\ket0_M\).  Since \(T\) is
finite, the predicate \(V(T)\) can be computed reversibly using finitely many
degenerate ancillas:
\[
\ket{T}\ket0_M\ket{0\cdots0}_{A_{\rm anc}}
\mapsto
\ket{T}\ket{V(T)}_M\ket{0\cdots0}_{A_{\rm anc}}.
\]
Then apply
\[
U_{\rm fb}
=
\ket1\!\bra1_M\otimes \mathrm{SWAP}_{FW}
+
\ket0\!\bra0_M\otimes \Id_{FW}.
\]
This unitary is energy preserving because \(M\) is degenerate and
\(\mathrm{SWAP}_{FW}\) commutes with \(H_F+H_W\).  If \(V(T)=1\), the fuel
excitation is moved to the battery and \(W_{\rm bat}=\Delta\).  If \(V(T)=0\),
the identity branch acts and \(W_{\rm bat}=0\).  Therefore
\[
W_{\rm bat}=\Delta V(T).
\]
Taking expectations gives
\[
\bbE[W_{\rm bat}]=\Delta\,\bbP[V(T)=1].
\]
\end{proof}

\begin{remark}
This proposition shows that the battery mechanism itself is not special to XOR
games.  The reason to use XOR games in the main text is that their local,
quantum, and nonsignalling values are standard and yield transparent
post-quantum thresholds.
\end{remark}

\section{Details for the chained Bell games}
\label{app:chained_details}

This appendix gives the elementary parts of the chained-game analysis.  The
quantum value is quoted from the standard chained Tsirelson bound
\cite{BraunsteinCaves1990,Wehner2006}.

\subsection{Classical value}

For the chained game \(\mathcal G_N\), the \(2N\) tested constraints are
\[
\alpha_j\oplus\beta_j=0,
\qquad
j=0,\ldots,N-1,
\]
\[
\alpha_{j+1}\oplus\beta_j=0,
\qquad
j=0,\ldots,N-2,
\]
and
\[
\alpha_0\oplus\beta_{N-1}=1.
\]
Here \(\alpha_j\) and \(\beta_j\) are deterministic local outputs.

A deterministic strategy cannot satisfy all constraints.  Indeed, from
\[
\alpha_j\oplus\beta_j=0
\]
we get
\[
\alpha_j=\beta_j
\]
for all \(j\).  From
\[
\alpha_{j+1}\oplus\beta_j=0
\]
for \(j=0,\ldots,N-2\), we get
\[
\alpha_{j+1}=\beta_j=\alpha_j.
\]
Therefore
\[
\alpha_0=\alpha_1=\cdots=\alpha_{N-1}.
\]
Since
\[
\beta_{N-1}=\alpha_{N-1},
\]
we obtain
\[
\alpha_0\oplus\beta_{N-1}=0,
\]
contradicting the final condition
\[
\alpha_0\oplus\beta_{N-1}=1.
\]
Thus at least one of the \(2N\) constraints must fail.

This upper bound is tight: choosing all deterministic outputs to be zero wins
all equality constraints and loses only the final inequality constraint.
Therefore
\[
\omega_{\mathsf L}(\mathcal G_N)=1-\frac{1}{2N}.
\]

\subsection{Nonsignalling value}

For every allowed input pair \((u,v)\), define
\[
P(a,b|u,v)
=
\begin{cases}
\frac12, & a\oplus b=f(u,v),\\
0, & a\oplus b\ne f(u,v).
\end{cases}
\]
Then the winning condition is satisfied with probability one.  Alice's marginal
is uniform:
\[
\sum_bP(a,b|u,v)=\frac12
\]
for both \(a=0,1\), independently of \(v\).  Bob's marginal is also uniform:
\[
\sum_aP(a,b|u,v)=\frac12
\]
for both \(b=0,1\), independently of \(u\).  Hence the behaviour is
nonsignalling on the support of the game.  It can be extended to all input pairs
with uniform marginals.  Therefore
\[
\omega_{\mathsf{NS}}(\mathcal G_N)=1.
\]

\subsection{Quantum value}

The quantum value is the standard chained Tsirelson value:
\[
\omega_{\mathsf Q}(\mathcal G_N)
=
\cos^2\left(\frac{\pi}{4N}\right).
\]
For \(N=2\), this gives
\[
\omega_{\mathsf Q}(\mathcal G_2)=\cos^2\left(\frac{\pi}{8}\right),
\]
which is the usual CHSH quantum winning probability.

\section{Convex-content bounds from battery data}
\label{app:content_bounds}

The battery value can also be used to lower-bound the fraction of a behaviour
that must lie outside a chosen resource class.

Let
\[
\mathsf C\subseteq \mathsf D
\]
be two behaviour classes, and assume
\[
\omega_{\mathsf D}(\mathcal G)>\omega_{\mathsf C}(\mathcal G).
\]
Suppose that
\[
P=(1-q)P_{\mathsf C}+qP_{\mathsf D},
\]
where
\[
P_{\mathsf C}\in\mathsf C,
\qquad
P_{\mathsf D}\in\mathsf D.
\]

\begin{proposition}[Lower bound on non-\(\mathsf C\) content]
\label{prop:content_bound}
If
\[
p=p_{\rm succ}^{\mathcal G}(P),
\]
then
\[
q
\ge
\frac{p-\omega_{\mathsf C}(\mathcal G)}
{\omega_{\mathsf D}(\mathcal G)-\omega_{\mathsf C}(\mathcal G)}
\]
whenever the numerator is positive.  In battery units,
\[
q
\ge
\frac{\bbE[W_{\rm bat}]/\Delta-\omega_{\mathsf C}(\mathcal G)}
{\omega_{\mathsf D}(\mathcal G)-\omega_{\mathsf C}(\mathcal G)}.
\]
\end{proposition}

\begin{proof}
The success probability is affine in the behaviour:
\[
p_{\rm succ}^{\mathcal G}(P)
=
(1-q)p_{\rm succ}^{\mathcal G}(P_{\mathsf C})
+
q p_{\rm succ}^{\mathcal G}(P_{\mathsf D}).
\]
Using the class values,
\[
p_{\rm succ}^{\mathcal G}(P_{\mathsf C})\le\omega_{\mathsf C}(\mathcal G),
\]
and
\[
p_{\rm succ}^{\mathcal G}(P_{\mathsf D})\le\omega_{\mathsf D}(\mathcal G).
\]
Therefore
\begin{align}
p
&\le
(1-q)\omega_{\mathsf C}(\mathcal G)
+
q\omega_{\mathsf D}(\mathcal G)\\
&=
\omega_{\mathsf C}(\mathcal G)
+
q\left[
\omega_{\mathsf D}(\mathcal G)-\omega_{\mathsf C}(\mathcal G)
\right].
\end{align}
Rearranging gives the stated bound.  The battery form follows from
\[
p=\bbE[W_{\rm bat}]/\Delta.
\]
\end{proof}

\subsection{CHSH nonlocal and post-quantum content}

For CHSH,
\[
\omega_{\mathsf L}=\frac34,
\qquad
\omega_{\mathsf Q}=\cos^2\frac{\pi}{8},
\qquad
\omega_{\mathsf{NS}}=1.
\]
A lower bound on nonsignalling nonlocal content is obtained by taking
\[
\mathsf C=\mathsf L,
\qquad
\mathsf D=\mathsf{NS}.
\]
Then
\[
q_{\rm NL}
\ge
\frac{\bbE[W_{\rm bat}]/\Delta-\frac34}{1-\frac34}
=
4\frac{\bbE[W_{\rm bat}]}{\Delta}-3.
\]
Using
\[
\frac{\bbE[W_{\rm bat}]}{\Delta}
=
\frac12+\frac{S}{8},
\]
this becomes
\[
q_{\rm NL}
\ge
\frac{S-2}{2}.
\]

For post-quantum content, take
\[
\mathsf C=\mathsf Q,
\qquad
\mathsf D=\mathsf{NS}.
\]
Then
\[
q_{\rm postQ}
\ge
\frac{\bbE[W_{\rm bat}]/\Delta-\cos^2(\pi/8)}
{1-\cos^2(\pi/8)}.
\]
In terms of the CHSH value,
\[
q_{\rm postQ}
\ge
\frac{S-2\sqrt2}{4-2\sqrt2}.
\]

\section{CHSH monogamy in battery form}
\label{app:monogamy}

For a tripartite nonsignalling behaviour \(P(a,b,c|x,y,z)\), let \(S_{AB}\) be
the CHSH value of the marginal shared by Alice and Bob, and let \(S_{AC}\) be
the CHSH value of the marginal shared by Alice and Charlie.  The standard CHSH
monogamy relation gives \cite{TonerVerstraete2006}
\[
S_{AB}+S_{AC}\le4,
\]
for consistently oriented CHSH expressions.

For the battery transducer applied to the two marginals,
\[
\bbE[W_{AB}]
=
\Delta\left(\frac12+\frac{S_{AB}}8\right),
\]
and
\[
\bbE[W_{AC}]
=
\Delta\left(\frac12+\frac{S_{AC}}8\right).
\]
Adding these identities gives
\[
\bbE[W_{AB}]+\bbE[W_{AC}]
=
\Delta\left[
1+\frac{S_{AB}+S_{AC}}8
\right].
\]
Using
\[
S_{AB}+S_{AC}\le4,
\]
we obtain
\[
\boxed{
\bbE[W_{AB}]+\bbE[W_{AC}]
\le
\frac32\Delta.
}
\]
Thus the usual nonsignalling CHSH monogamy relation becomes a battery-monogamy
bound.

\section{Exact binomial confidence intervals}
\label{app:binomial_intervals}

The main text uses Hoeffding and Azuma--Hoeffding bounds because they are simple
and analytic.  In the i.i.d. setting, the work bits are Bernoulli random
variables, so one can also use exact or sharper binomial intervals.

Let
\[
Z_i=\frac{W_i}{\Delta}\in\{0,1\},
\]
and let
\[
k=\sum_{i=1}^n Z_i,
\qquad
\hat p=\frac{k}{n}.
\]
In the ideal i.i.d. transducer model,
\[
k\sim {\rm Binomial}(n,p),
\]
where
\[
p=p_{\rm succ}^{\mathcal G}(P).
\]

\subsection{Clopper--Pearson interval}

A two-sided Clopper--Pearson interval \cite{ClopperPearson1934} with error
probability \(\alpha\) is
\[
[p_L,p_U],
\]
where, for \(0<k<n\),
\[
p_L=
{\rm Beta}^{-1}\left(\frac{\alpha}{2};\,k,\,n-k+1\right),
\]
and
\[
p_U=
{\rm Beta}^{-1}\left(1-\frac{\alpha}{2};\,k+1,\,n-k\right).
\]
Here \({\rm Beta}^{-1}(q;a,b)\) is the \(q\)-quantile of the beta distribution
with parameters \(a,b\).  The endpoint conventions are
\[
p_L=0 \quad\text{if } k=0,
\]
and
\[
p_U=1 \quad\text{if } k=n.
\]

A one-sided lower confidence bound is obtained by replacing \(\alpha/2\) by
\(\alpha\):
\[
p_L^{(1)}
=
{\rm Beta}^{-1}\left(\alpha;\,k,\,n-k+1\right).
\]

\subsection{Wilson interval}

The Wilson interval \cite{Wilson1927} is often shorter while maintaining good
coverage.  Let
\[
z=\Phi^{-1}\left(1-\frac{\alpha}{2}\right),
\]
where \(\Phi\) is the standard normal cumulative distribution function.  The
Wilson interval is
\[
\left[
\frac{
\hat p+\frac{z^2}{2n}
-
z\sqrt{\frac{\hat p(1-\hat p)}{n}+\frac{z^2}{4n^2}}
}{
1+\frac{z^2}{n}
},
\,
\frac{
\hat p+\frac{z^2}{2n}
+
z\sqrt{\frac{\hat p(1-\hat p)}{n}+\frac{z^2}{4n^2}}
}{
1+\frac{z^2}{n}
}
\right].
\]

\subsection{Mapping to CHSH}

Any confidence interval
\[
p\in[p_L,p_U]
\]
gives a battery interval
\[
\bbE[W_{\rm bat}]
\in
[\Delta p_L,\Delta p_U].
\]
For CHSH,
\[
S=8\left(p-\frac12\right).
\]
Thus
\[
S\in
\left[
8\left(p_L-\frac12\right),
8\left(p_U-\frac12\right)
\right].
\]
A finite-data post-quantumness certificate is obtained whenever
\[
p_L>\cos^2\frac{\pi}{8},
\]
or equivalently
\[
8\left(p_L-\frac12\right)>2\sqrt2.
\]

\section{Memory reset variants}
\label{app:memory_variants}

The Landauer term in the main text refers to a compressed success/failure
memory.  Here we spell out how the cost depends on what is stored.

\subsection{Compressed success memory}

If the only persistent memory is
\[
Z=\mathbf{1}\{\mathrm{win}\},
\]
then
\[
\bbP[Z=1]=p,
\qquad
\bbP[Z=0]=1-p,
\]
where
\[
p=p_{\rm succ}^{\mathcal G}(P).
\]
The memory entropy is
\[
H(Z)=h_2(p).
\]
Blind erasure costs at least
\[
Q_{\rm reset}
\ge
\kB T\ln2\,h_2(p).
\]

\subsection{Full transcript memory}

If the implementation stores the full transcript
\[
T=(U,V,R,A,B),
\]
then the erasure cost is governed by \(H(T)\), not merely by \(h_2(p)\).
Since \(Z\) is a deterministic function of \(T\),
\[
H(T)\ge H(Z)=h_2(p).
\]
Thus erasing the full transcript is at least as costly as erasing the compressed
success/failure bit.

\subsection{Side-information-assisted reset}

If the erasing agent has side information \(Y\) correlated with the memory, then
the relevant classical entropy can be reduced to a conditional entropy
\(H(Z|Y)\).  The present work deliberately uses blind reset of the persistent
local memory, so such side-information-assisted reductions are not used.

\subsection{Reversible uncomputation}

In the reversible-controller implementation, the success bit is computed, used,
and uncomputed.  No persistent \(Z\) remains.  Therefore no Landauer erasure
cost is assigned to the success bit in that implementation.

\section{Detailed fuel-battery balance}
\label{app:fuel_balance}

The initial fuel-battery state is
\[
\ket{1}_F\ket{0}_W.
\]
The initial fuel energy is
\[
E_F^{\rm in}=\Delta,
\]
and the initial battery energy is
\[
E_W^{\rm in}=0.
\]

If the game is won, the SWAP branch gives
\[
\ket{1}_F\ket{0}_W
\longmapsto
\ket{0}_F\ket{1}_W.
\]
Thus
\[
E_F^{\rm out}=0,
\qquad
E_W^{\rm out}=\Delta,
\]
and
\[
\Delta E_F=-\Delta,
\qquad
\Delta E_W=+\Delta.
\]

If the game is lost, the identity branch gives
\[
\ket{1}_F\ket{0}_W
\longmapsto
\ket{1}_F\ket{0}_W.
\]
Thus
\[
E_F^{\rm out}=\Delta,
\qquad
E_W^{\rm out}=0,
\]
and
\[
\Delta E_F=0,
\qquad
\Delta E_W=0.
\]

Combining the two cases,
\[
\Delta E_W
=
\Delta\,\mathbf{1}\{\mathrm{win}\},
\]
and
\[
\Delta E_F
=
-\Delta\,\mathbf{1}\{\mathrm{win}\}.
\]
Therefore
\[
\Delta E_F+\Delta E_W=0
\]
in every run.

Averaging gives
\[
\bbE[\Delta E_W]=\Delta p,
\]
and
\[
\bbE[\Delta E_F]=-\Delta p.
\]
Thus restoring the fuel to \(\ket1_F\) costs at least
\[
\Delta p
\]
on average in the ideal energy-eigenstate model.  This is why the minimal
cyclic fuel-restoration cost is success-weighted rather than automatically
equal to \(\Delta\) per attempt.

\bibliographystyle{apsrev4-2}
\bibliography{refs}

@article{LostaglioArxiv1410,
  title = {Quantum Coherence, Time-Translation Symmetry, and Thermodynamics},
  author = {Lostaglio, Matteo and Korzekwa, Kamil and Jennings, David and Rudolph, Terry},
  journal = {Phys. Rev. X},
  volume = {5},
  issue = {2},
  pages = {021001},
  numpages = {11},
  year = {2015},
  month = {Apr},
  publisher = {American Physical Society},
  doi = {10.1103/PhysRevX.5.021001},
  url = {https://link.aps.org/doi/10.1103/PhysRevX.5.021001}
}

@article{LostaglioNatComm2015,
  author  = {Lostaglio, Matteo and Jennings, David and Rudolph, Terry},
  title   = {Description of quantum coherence in thermodynamic processes requires constraints beyond free energy},
  journal = {Nature Communications},
  volume  = {6},
  pages   = {6383},
  year    = {2015},
  doi     = {10.1038/ncomms7383}
}

@article{BrandaoPNAS2015,
  author  = {Brand{\~a}o, Fernando G. S. L. and Horodecki, Micha{\l} and Ng, Nelly and Oppenheim, Jonathan and Wehner, Stephanie},
  title   = {The second laws of quantum thermodynamics},
  journal = {Proceedings of the National Academy of Sciences},
  volume  = {112},
  number  = {11},
  pages   = {3275--3279},
  year    = {2015},
  doi     = {10.1073/pnas.1411728112}
}

@article{HorodeckiNatComm2013,
  author  = {Horodecki, Micha{\l} and Oppenheim, Jonathan},
  title   = {Fundamental limitations for quantum and nanoscale thermodynamics},
  journal = {Nature Communications},
  volume  = {4},
  pages   = {2059},
  year    = {2013},
  doi     = {10.1038/ncomms3059}
}

@article{PerryPRX2018,
  author  = {Perry, Christopher and {\'C}wikli{\'n}ski, Piotr and Anders, Janet and Horodecki, Micha{\l} and Oppenheim, Jonathan},
  title   = {A Sufficient Set of Experimentally Implementable Thermal Operations for Small Systems},
  journal = {Physical Review X},
  volume  = {8},
  pages   = {041049},
  year    = {2018},
  doi     = {10.1103/PhysRevX.8.041049}
}

@article{ETO2018,
  author  = {Lostaglio, Matteo and Alhambra, {\'A}lvaro M. and Perry, Christopher},
  title   = {Elementary Thermal Operations},
  journal = {Quantum},
  volume  = {2},
  pages   = {52},
  year    = {2018},
  doi     = {10.22331/q-2018-02-08-52}
}

@article{SagawaUeda2008,
  author  = {Sagawa, Takahiro and Ueda, Masahito},
  title   = {Second Law of Thermodynamics with Discrete Quantum Feedback Control},
  journal = {Physical Review Letters},
  volume  = {100},
  pages   = {080403},
  year    = {2008},
  doi     = {10.1103/PhysRevLett.100.080403}
}

@article{SagawaUeda2010,
  author  = {Sagawa, Takahiro and Ueda, Masahito},
  title   = {Generalized Jarzynski Equality under Nonequilibrium Feedback Control},
  journal = {Physical Review E},
  volume  = {82},
  pages   = {021101},
  year    = {2010},
  doi     = {10.1103/PhysRevE.82.021101}
}

@article{PopescuRohrlich1994,
  author  = {Popescu, Sandu and Rohrlich, Daniel},
  title   = {Quantum nonlocality as an axiom},
  journal = {Foundations of Physics},
  volume  = {24},
  number  = {3},
  pages   = {379--385},
  year    = {1994},
  doi     = {10.1007/BF02058098}
}

@article{CHSH1969,
  author  = {Clauser, John F. and Horne, Michael A. and Shimony, Abner and Holt, Richard A.},
  title   = {Proposed Experiment to Test Local Hidden-Variable Theories},
  journal = {Physical Review Letters},
  volume  = {23},
  pages   = {880--884},
  year    = {1969},
  doi     = {10.1103/PhysRevLett.23.880}
}

@article{Wilson1927,
  author  = {Wilson, Edwin B.},
  title   = {Probable Inference, the Law of Succession, and Statistical Inference},
  journal = {Journal of the American Statistical Association},
  volume  = {22},
  number  = {158},
  pages   = {209--212},
  year    = {1927},
  doi     = {10.1080/01621459.1927.10502953}
}

@article{ClopperPearson1934,
  author  = {Clopper, Charles J. and Pearson, Egon S.},
  title   = {The Use of Confidence or Fiducial Limits Illustrated in the Case of the Binomial},
  journal = {Biometrika},
  volume  = {26},
  number  = {4},
  pages   = {404--413},
  year    = {1934},
  doi     = {10.1093/biomet/26.4.404}
}

@article{Tsirelson1980,
  author  = {Tsirelson, Boris S.},
  title   = {Quantum generalizations of Bell's inequality},
  journal = {Letters in Mathematical Physics},
  volume  = {4},
  number  = {2},
  pages   = {93--100},
  year    = {1980},
  doi     = {10.1007/BF00417500}
}

@article{Landauer1961,
  author  = {Landauer, Rolf},
  title   = {Irreversibility and Heat Generation in the Computing Process},
  journal = {IBM Journal of Research and Development},
  volume  = {5},
  number  = {3},
  pages   = {183--191},
  year    = {1961},
  doi     = {10.1147/rd.53.0183}
}

@article{Bennett1982,
  author  = {Bennett, Charles H.},
  title   = {The thermodynamics of computation---a review},
  journal = {International Journal of Theoretical Physics},
  volume  = {21},
  number  = {12},
  pages   = {905--940},
  year    = {1982},
  doi     = {10.1007/BF02084158}
}

@article{BarrettPRA2005,
  author  = {Barrett, Jonathan and Linden, Noah and Massar, Serge and Pironio, Stefano and Popescu, Sandu and Roberts, David},
  title   = {Nonlocal correlations as an information-theoretic resource},
  journal = {Physical Review A},
  volume  = {71},
  pages   = {022101},
  year    = {2005},
  doi     = {10.1103/PhysRevA.71.022101}
}

@article{ReebWolf2014,
  title={An improved Landauer principle with finite-size corrections},
  author={Reeb, David and Wolf, Michael M.},
  journal={New Journal of Physics},
  volume={16},
  pages={103011},
  year={2014},
  doi={10.1088/1367-2630/16/10/103011}
}

@article{Goold2016,
  title={The role of quantum information in thermodynamics---a topical review},
  author={Goold, John and Huber, Marcus and Riera, Arnau and del Rio, Lidia and Skrzypczyk, Paul},
  journal={Journal of Physics A: Mathematical and Theoretical},
  volume={49},
  number={14},
  pages={143001},
  year={2016},
  doi={10.1088/1751-8113/49/14/143001}
}

@misc{TonerVerstraete2006,
      title={Monogamy of Bell correlations and Tsirelson's bound}, 
      author={Benjamin Toner and Frank Verstraete},
      year={2006},
      eprint={quant-ph/0611001},
      archivePrefix={arXiv},
      primaryClass={quant-ph},
      url={https://arxiv.org/abs/quant-ph/0611001}, 
}

@article{Parrondo2015,
  author  = {Parrondo, Juan M. R. and Horowitz, Jordan M. and Sagawa, Takahiro},
  title   = {Thermodynamics of information},
  journal = {Nature Physics},
  volume  = {11},
  pages   = {131--139},
  year    = {2015},
  doi     = {10.1038/nphys3230}
}

@inproceedings{CleveHoyerTonerWatrous2004,
  author        = {Cleve, Richard and H{\o}yer, Peter and Toner, Benjamin and Watrous, John},
  title         = {Consequences and Limits of Nonlocal Strategies},
  booktitle     = {Proceedings of the 19th IEEE Annual Conference on Computational Complexity},
  pages         = {236--249},
  year          = {2004},
  publisher     = {IEEE Computer Society},
  doi           = {10.1109/CCC.2004.1313847},
  eprint        = {quant-ph/0404076},
  archivePrefix = {arXiv},
  primaryClass  = {quant-ph}
}

@article{BraunsteinCaves1990,
  author  = {Braunstein, Samuel L. and Caves, Carlton M.},
  title   = {Wringing out better {Bell} inequalities},
  journal = {Annals of Physics},
  volume  = {202},
  number  = {1},
  pages   = {22--56},
  year    = {1990},
  doi     = {10.1016/0003-4916(90)90339-P}
}

@article{Wehner2006,
  author  = {Wehner, Stephanie},
  title   = {Tsirelson bounds for generalized {Clauser-Horne-Shimony-Holt} inequalities},
  journal = {Physical Review A},
  volume  = {73},
  pages   = {022110},
  year    = {2006},
  doi     = {10.1103/PhysRevA.73.022110},
  eprint  = {quant-ph/0510076},
  archivePrefix = {arXiv}
}

@misc{cwiklinski2026_Szilard,
      title={Thermodynamic Value of XOR-Game-Induced Side Information in a Szilard Engine}, 
      author={Piotr Ćwikliński},
      year={2026},
      eprint={2605.12044},
      archivePrefix={arXiv},
      primaryClass={quant-ph},
      url={https://arxiv.org/abs/2605.12044}, 
}

@article{Hoeffding1963,
  author  = {Hoeffding, Wassily},
  title   = {Probability Inequalities for Sums of Bounded Random Variables},
  journal = {Journal of the American Statistical Association},
  volume  = {58},
  number  = {301},
  pages   = {13--30},
  year    = {1963},
  doi     = {10.1080/01621459.1963.10500830}
}

@article{Azuma1967,
  author  = {Azuma, Kazuoki},
  title   = {Weighted Sums of Certain Dependent Random Variables},
  journal = {Tohoku Mathematical Journal},
  volume  = {19},
  number  = {3},
  pages   = {357--367},
  year    = {1967},
  doi     = {10.2748/tmj/1178243286}
}

@misc{RoutRavichandranHorodeckiChaturvedi2026,
      title={Quantum work beyond classical (commuting) limits}, 
      author={Sumit Rout and Aravinth Balaji Ravichandran and Paweł Horodecki and Anubhav Chaturvedi},
      year={2026},
      eprint={2605.04021},
      archivePrefix={arXiv},
      primaryClass={quant-ph},
      url={https://arxiv.org/abs/2605.04021}, 
}

@article{Zhi2-25_quantumbatteries,
  title = {Topological Quantum Batteries},
  author = {Lu, Zhi-Guang and Tian, Guoqing and L\"u, Xin-You and Shang, Cheng},
  journal = {Phys. Rev. Lett.},
  volume = {134},
  issue = {18},
  pages = {180401},
  numpages = {8},
  year = {2025},
  month = {May},
  publisher = {American Physical Society},
  doi = {10.1103/PhysRevLett.134.180401},
  url = {https://link.aps.org/doi/10.1103/PhysRevLett.134.180401}
}

@article{Chaki2025_distillationfromquantumbatteries,
  title = {Auxiliary-assisted energy distillation from quantum batteries},
  author = {Chaki, Paranjoy and Bhattacharyya, Aparajita and Sen, Kornikar and Sen, Ujjwal},
  journal = {Phys. Rev. A},
  volume = {112},
  issue = {5},
  pages = {052446},
  numpages = {16},
  year = {2025},
  month = {Nov},
  publisher = {American Physical Society},
  doi = {10.1103/cyrc-ms34},
  url = {https://link.aps.org/doi/10.1103/cyrc-ms34}
}

@article{Chaki2026_assistantsquantumbatteries,
doi = {10.1088/1751-8121/ae3ff5},
url = {https://doi.org/10.1088/1751-8121/ae3ff5},
year = {2026},
month = {feb},
publisher = {IOP Publishing},
volume = {59},
number = {6},
pages = {065303},
author = {Chaki, Paranjoy and Bhattacharyya, Aparajita and Sen, Kornikar and Sen, Ujjwal},
title = {Role of energy-invariant assistants in energy extraction from quantum batteries},
journal = {Journal of Physics A: Mathematical and Theoretical}
}

@misc{chaki2025positivenonpositivemeasurementsenergy,
      title={Positive and non-positive measurements in energy extraction from quantum batteries}, 
      author={Paranjoy Chaki and Aparajita Bhattacharyya and Kornikar Sen and Ujjwal Sen},
      year={2025},
      eprint={2404.18745},
      archivePrefix={arXiv},
      primaryClass={quant-ph},
      url={https://arxiv.org/abs/2404.18745}, 
}

@article{LostaglioReview2019,
  author = {Lostaglio, Matteo},
  title = {An Introductory Review of the Resource Theory Approach to Thermodynamics},
  journal = {Reports on Progress in Physics},
  volume = {82},
  pages = {114001},
  year = {2019},
  doi = {10.1088/1361-6633/ab46e5}
}

@article{ChitambarGour2019,
  author = {Chitambar, Eric and Gour, Gilad},
  title = {Quantum Resource Theories},
  journal = {Reviews of Modern Physics},
  volume = {91},
  pages = {025001},
  year = {2019},
  doi = {10.1103/RevModPhys.91.025001}
}

@article{Gour2015,
  author = {Gour, Gilad and M{\"u}ller, Markus P. and Narasimhachar, Varun and Spekkens, Robert W. and Halpern, Nicole Yunger},
  title = {The Resource Theory of Informational Nonequilibrium in Thermodynamics},
  journal = {Physics Reports},
  volume = {583},
  pages = {1--58},
  year = {2015},
  doi = {10.1016/j.physrep.2015.04.003}
}

@article{Cwiklinski2015,
  author = {{\'C}wikli{\'n}ski, Piotr and Studzi{\'n}ski, Micha{\l} and Horodecki, Micha{\l} and Oppenheim, Jonathan},
  title = {Limitations on the Evolution of Quantum Coherences: Towards Fully Quantum Second Laws of Thermodynamics},
  journal = {Physical Review Letters},
  volume = {115},
  pages = {210403},
  year = {2015},
  doi = {10.1103/PhysRevLett.115.210403}
}

@article{Shiraishi2025,
  author  = {Shiraishi, Naoto},
  title   = {Quantum Thermodynamics with Coherence: Covariant Gibbs-Preserving Operation Is Characterized by the Free Energy},
  journal = {Physical Review Letters},
  volume  = {134},
  number  = {16},
  pages   = {160402},
  year    = {2025},
  doi     = {10.1103/PhysRevLett.134.160402},
  eprint  = {2406.06234},
  archivePrefix = {arXiv},
  primaryClass  = {quant-ph}
}

@article{SonNg2023,
doi = {10.1088/2058-9565/ad7ef5},
url = {https://doi.org/10.1088/2058-9565/ad7ef5},
year = {2024},
month = {oct},
publisher = {IOP Publishing},
volume = {10},
number = {1},
pages = {015011},
author = {Son, Jeongrak and Ng, Nelly H Y},
title = {A hierarchy of thermal processes collapses under catalysis},
journal = {Quantum Science and Technology}
}

@article{HuDing2019,
  author  = {Hu, Xueyuan and Ding, Feng},
  title   = {Thermal operations involving a single-mode bosonic bath},
  journal = {Physical Review A},
  volume  = {99},
  number  = {1},
  pages   = {012104},
  year    = {2019},
  doi     = {10.1103/PhysRevA.99.012104},
  eprint  = {1807.02942},
  archivePrefix = {arXiv},
  primaryClass  = {quant-ph}
}

@article{HackMendl2025,
  author  = {Hack, Pedro and Mendl, Christian B.},
  title   = {Universality and classification of elementary thermal operations},
  journal = {Journal of Physics A: Mathematical and Theoretical},
  volume  = {58},
  number  = {31},
  pages   = {315302},
  year    = {2025},
  doi     = {10.1088/1751-8121/adf26e},
  eprint  = {2312.11223},
  archivePrefix = {arXiv},
  primaryClass  = {quant-ph}
}

@article{KorzekwaWorkFromCoh2016,
  author  = {Korzekwa, Kamil and Lostaglio, Matteo and Oppenheim, Jonathan and Jennings, David},
  title   = {The extraction of work from quantum coherence},
  journal = {New Journal of Physics},
  year    = {2016},
  volume  = {18},
  pages   = {023045},
  doi     = {10.1088/1367-2630/18/2/023045}
}

@misc{Su2025_thermoresource, author = {Su, Shanhe and Fu, Cong and Pan, Ousi and Xia, Shihao and Liu, Fei and Chen, Jincan}, title = {Quantum Coherence as a Thermodynamic Resource Beyond the Classical Uncertainty Bound}, year = {2025}, eprint = {2510.20873}, archivePrefix = {arXiv}, primaryClass = {quant-ph}, doi = {10.48550/arXiv.2510.20873}, url = {https://arxiv.org/abs/2510.20873} }

@article{CaravelliQuantum2021,
  author  = {Caravelli, Francesco and Yan, Bin and Garc{\'i}a-Pintos, Luis Pedro and Hamma, Alioscia},
  title   = {Energy storage and coherence in closed and open quantum batteries},
  journal = {Quantum},
  year    = {2021},
  volume  = {5},
  pages   = {505},
  doi     = {10.22331/q-2021-07-15-505},
  url     = {https://quantum-journal.org/papers/q-2021-07-15-505/}
}

@article{Mazurek_2018,
  title = {Decomposability and convex structure of thermal processes},
  author = {Mazurek, Paweł and Horodecki, Michał},
  year = {2018},
  journal = {New Journal of Physics},
  volume = {20},
  number = {5},
  pages = {053040},
  doi = {10.1088/1367-2630/aac057}
}

@article{Faist2015,
  title = {Gibbs-preserving maps outperform thermal operations in the quantum regime},
  author = {Faist, Philippe and Oppenheim, Jonathan and Renner, Renato},
  journal = {New J. Phys.},
  volume = {17},
  pages = {043003},
  year = {2015},
  doi = {10.1088/1367-2630/17/4/043003}
}

@article{Lostaglio2019RoPP,
  doi = {10.1088/1361-6633/ab46e5},
  year = {2019},
  journal = {Reports on Progress in Physics},
  volume = {82},
  number = {11},
  pages = {114001},
  author = {Lostaglio, Matteo},
  title = {An introductory review of the resource theory approach to thermodynamics}
}

@article{Renes2014,
  title = {Work cost of thermal operations in quantum thermodynamics},
  author = {Renes, Joseph M.},
  journal = {Eur. Phys. J. Plus},
  volume = {129},
  pages = {153},
  year = {2014},
  doi = {10.1140/epjp/i2014-14153-8}
}

@article{Rodriguez_2024_Optimal_Control,
  title = {Optimal quantum control of charging quantum batteries},
  author = {Rodriguez, R R and Ahmadi, B and Suarez, G and Mazurek, P and Barzanjeh, S and Horodecki, P},
  journal = {New Journal of Physics},
  volume = {26},
  number = {4},
  pages = {043004},
  year = {2024},
  publisher = {IOP Publishing},
  doi = {10.1088/1367-2630/ad3843},
  url = {https://doi.org}
}

@article{LipkaBartosik2021secondlawof,
  doi = {10.22331/q-2021-03-10-408},
  url = {https://doi.org/10.22331/q-2021-03-10-408},
  title = {Second law of thermodynamics for batteries with vacuum state},
  author = {Lipka-Bartosik, Patryk and Mazurek, Pawe{\l{}} and Horodecki, Micha{\l{}}},
  journal = {{Quantum}},
  issn = {2521-327X},
  publisher = {{Verein zur F{\"{o}}rderung des Open Access Publizierens in den Quantenwissenschaften}},
  volume = {5},
  pages = {408},
  month = mar,
  year = {2021}
}

@article{Borhan2025_batteries,
  title = {Superoptimal charging of quantum batteries via reservoir engineering: Arbitrary energy transfer unlocked},
  author = {Ahmadi, Borhan and Mazurek, Pawe\l{} and Barzanjeh, Shabir and Horodecki, Pawe\l{}},
  journal = {Phys. Rev. Appl.},
  volume = {23},
  issue = {2},
  pages = {024010},
  numpages = {14},
  year = {2025},
  month = {Feb},
  publisher = {American Physical Society},
  doi = {10.1103/PhysRevApplied.23.024010},
  url = {https://link.aps.org/doi/10.1103/PhysRevApplied.23.024010}
}

@article{Malavazi2025ChargepreservingOI,
  title={Charge-preserving operations in quantum batteries},
  author={Andr{\'e} H. A. Malavazi and Borhan Ahmadi and Pawe{\l} Horodecki and Pedro R. Dieguez},
  journal={PRX Energy},
  year={2025},
  volume={4},
  pages={043003},
  doi={10.1103/PRXEnergy.4.043003},
  url={https://aps.org}
}

@article{GoriniKossakowskiSudarshan1976,
  author  = {Gorini, Vittorio and Kossakowski, Andrzej and Sudarshan, E. C. G.},
  title   = {Completely positive dynamical semigroups of {N}-level systems},
  journal = {Journal of Mathematical Physics},
  volume  = {17},
  pages   = {821-825},
  year    = {1976},
  doi     = {10.1063/1.522979}
}

@article{Lindblad1976,
  author  = {Lindblad, Goran},
  title   = {On the generators of quantum dynamical semigroups},
  journal = {Communications in Mathematical Physics},
  volume  = {48},
  pages   = {119-130},
  year    = {1976},
  doi     = {10.1007/BF01608499}
}

@article{PuszWoronowicz1978,
  author  = {Pusz, Wies{\l}aw and Woronowicz, Stanis{\l}aw L.},
  title   = {Passive states and {KMS} states for general quantum systems},
  journal = {Communications in Mathematical Physics},
  volume  = {58},
  pages   = {273-290},
  year    = {1978},
  doi     = {10.1007/BF01614224}
}

@article{AllahverdyanBalianNieuwenhuizen2004,
  author  = {A. E. Allahverdyan and R. Balian and Th. M. Nieuwenhuizen},
  title   = {Maximal work extraction from finite quantum systems},
  journal = {Europhysics Letters},
  volume  = {67},
  number  = {4},
  pages   = {565},
  year    = {2004},
  doi     = {10.1209/epl/i2004-10101-2}
}

@article{deOliveiraJuniorBraskLipkaBartosik2025,
 title = {Heat as a Witness of Quantum Properties},
  author = {de Oliveira Junior, A. and Brask, Jonatan Bohr and Lipka-Bartosik, Patryk},
  journal = {Phys. Rev. Lett.},
  volume = {134},
  issue = {5},
  pages = {050401},
  numpages = {7},
  year = {2025},
  month = {Feb},
  publisher = {American Physical Society},
  doi = {10.1103/PhysRevLett.134.050401},
  url = {https://link.aps.org/doi/10.1103/PhysRevLett.134.050401}
}
\end{document}